\documentclass[11pt,a4paper]{article}

\usepackage{amsmath,amssymb,amsthm}
\usepackage{cite}
\usepackage{booktabs}
\usepackage{array}
\usepackage{graphicx}
\usepackage{url}
\usepackage{float}
\usepackage{geometry}
\usepackage{tikz}
\usepackage{pgfplots}
\usetikzlibrary{arrows.meta,calc}
\usepgfplotslibrary{fillbetween}
\pgfplotsset{compat=1.17}

\allowdisplaybreaks
\newtheorem{theorem}{Theorem}[section]
\newtheorem{proposition}[theorem]{Proposition}
\newtheorem{corollary}[theorem]{Corollary}
\newtheorem{lemma}[theorem]{Lemma}
\theoremstyle{definition}

\theoremstyle{remark}

\newcommand{\dd}{\,d}
\newcommand{\R}{\mathbb{R}}
\newcommand{\C}{\mathbb{C}}
\newcommand{\diag}{\mathrm{diag}}
\newcommand{\BH}{\mathrm{BH}}

\title{Collective Onset of Matter-Induced Scalarization around a Black Hole with Two Thin Shells}

\author{Masahiro Kaminaga\\
Department of Information Technology, Faculty of Engineering,\\
Tohoku Gakuin University, Sendai, Japan \\
E-mail: kaminaga@g.tohoku-gakuin.ac.jp \\
ORCID: 0000-0001-7204-8300}

\date{}

\begin{document}

\maketitle

\begin{abstract}

We study the linear onset of matter-induced scalarization for a static
black hole surrounded by two spatially separated thin matter shells.
We assume that the conformal matter coupling is unity on the
scalar-free background and that its logarithmic derivative vanishes
there. The scalar perturbation then decouples from the metric and
matter perturbations at first order, while the invariant surface
traces of the shells give singular terms in the scalar equation.

For a static spherical black-hole exterior with no bulk scalar
effective-mass term between the shells, the onset reduces to the
finite-rank condition
$$
1-a_1-a_2+(1-\chi)a_1a_2=0,
\qquad
\chi=\frac{{\cal S}(R_2)}{{\cal S}(R_1)},
$$
where $a_j$ is the attractive strength of shell $j$ normalized by its
one-shell threshold on the same background and ${\cal S}$ is the
static radial resistance.
This relation shows that two individually subcritical shells can
collectively destabilize the scalar-free black hole.

We derive the critical scalar cloud in closed form and verify
numerically, in the Schwarzschild probe problem, that a growing mode
with a finite growth rate appears beyond the threshold.
For an exact scalar-free background consisting of three Schwarzschild
regions joined by two Israel shells, we separately derive the static
onset condition and the corresponding critical cloud.
The exact Israel trace also exhibits a sign transition whose
light-shell limit occurs at the Schwarzschild photon sphere.

\end{abstract}


\section{Introduction}

Matter outside a black hole can change the scalar stability of a scalar-free solution.
In scalar-tensor gravity, spontaneous scalarization was first found for compact stars by Damour and Esposito-Far\`ese \cite{Damour1993,Damour1996}.
Its linear instability is described by modes with $\omega=i\kappa$, $\kappa>0$, and the onset is their zero-frequency limit \cite{Harada1997}.
A recent comprehensive review of scalarization mechanisms, compact objects, and observational restrictions is Ref.~\cite{Doneva2024}.
In the present paper, the word ``tachyonic'' is used only in this standard linear sense.
Under the usual assumptions, isolated stationary black holes in the standard massless scalar-tensor setting satisfy no-hair results \cite{Sotiriou2012}.
The conclusion changes when matter is present outside the horizon, because the trace of the matter stress tensor gives the scalar perturbation a position-dependent effective mass \cite{Cardoso2013PRD,Cardoso2013PRL}.
Recent studies have considered smooth dark-matter environments.
Examples include scalarization and superradiant instability in a halo
\cite{Tanaka2025}, scalar clouds induced by mass-varying dark matter
\cite{Sadjadi2026}, and scalarization in a conformally coupled
dark-matter halo \cite{GasemiZad2026}.
These smooth-profile calculations and the thin-shell model considered below address complementary questions: the former emphasize particular environmental profiles, whereas the latter isolates the nonlocal cooperation of two separated trace concentrations in an analytically solvable form.

Cardoso, Carucci, Pani, and Sotiriou showed that surrounding matter can induce black-hole scalarization, and they also constructed a self-consistent spherical black hole with one thin matter shell \cite{Cardoso2013PRD,Cardoso2013PRL}.
A thin shell can therefore represent either an approximation to a
smooth matter profile or part of a consistent scalarized configuration.
Laeuger et al. studied the ringdown of a black hole surrounded by a single thin matter shell, taking account of the coupled perturbations of the spacetime and the shell \cite{Laeuger2025}.
A complementary frequency-domain study of black-hole echoes used a controlled transfer-function model to separate the homogeneous resonance spectrum from source-dependent excitation effects \cite{KaminagaEcho2026}.

In general relativity, a surface layer is described by the Israel junction conditions, which relate its surface stress tensor to the jump of the extrinsic curvature \cite{Israel1966}.
Acu\~na-C\'ardenas, Sarbach, and Tessieri studied wave propagation on exact static spacetimes formed by joining Schwarzschild regions across concentric thin matter shells \cite{AcunaCardenas2024}.
Their problem concerns wave transmission through a prescribed shell spacetime.
Here, the invariant surface traces produce singular effective mass terms, and we determine the resulting collective tachyonic onset around a central black hole.

The Israel conditions determine the piecewise background geometry and the mechanical restrictions on a static shell, but they do not by themselves determine whether the scalar-free branch has a tachyonic scalar mode.
That question requires the linear scalar equation, whose singular coefficient is fixed by the surface trace.
Can two matter layers, neither of which reaches its own one-shell scalarization threshold, destabilize the scalar-free Schwarzschild black hole collectively?
The two-shell configuration is the minimal model in which nonlocal cooperation between separated matter regions can be isolated.
The model may be regarded as the thin-layer limit of a matter trace with two separated radial peaks, as would occur in a stratified configuration.
We do not repeat the one-shell nonlinear construction.
Instead, we determine the change in the onset threshold caused by scalar propagation between the two layers.

This question is a black-hole analog of collective scalarization.
Cardoso, Foschi, and Zilh\~ao showed in flat-space models, including
concentric thin-shell configurations, that several individually
subcritical components can produce a growing mode collectively, and
they also explained why a simple averaging picture can fail
\cite{Cardoso2020}.
In the present problem, the shells surround a black-hole horizon and
their scalar strengths are fixed by the invariant surface traces.
The horizon boundary condition and the Schwarzschild radial geometry
then modify the static scalar exchange between the separated shells.
We derive the collective threshold on a scalar-free two-shell
spacetime that satisfies the Israel junction conditions exactly.

In the thin-layer limit, the integrated surface trace of each matter layer appears as a delta interaction in the radial scalar equation.
The shell strength is fixed by the invariant surface trace and the proper-distance delta distribution, rather than introduced as an arbitrary potential parameter.
The zero-frequency onset condition then reduces to a finite-dimensional boundary matrix problem.
Delta interactions supported on spheres are standard solvable models in Schr\"odinger operator theory \cite{Albeverio2005,Antoine1987,Shabani1988,Behrndt2013}.
Earlier partial-wave scattering for finitely many concentric spheres was studied in Ref.~\cite{Hounkonnou1997}, while a three-dimensional boundary-operator formulation for concentric delta shells was developed in Ref.~\cite{Kaminaga2026}.
These works provide the boundary-matrix method used below.
They do not determine the gravitational shell strengths or the
black-hole threshold kernel.

The differences from the flat-space problem are the horizon boundary
condition, the logarithmic Schwarzschild resistance, and, beyond the
probe limit, the Israel junctions.

In the Schwarzschild probe problem, the collective onset condition is obtained in closed form as
$$
1-a_1-a_2+(1-\chi)a_1a_2=0,
$$
where
$$
\chi=\frac{\log f(R_2)}{\log f(R_1)},
\qquad
f(r)=1-\frac{2M}{r}.
$$
Here $a_j$ is the attractive strength of shell $j$ normalized by its one-shell threshold on the same Schwarzschild background.
In this one-shell diagnostic, the two gravitating shells are kept in the fixed background geometry, while only the $j$th scalar surface interaction is retained.
The number $\chi$ is the static geometric overlap of the two shells.
Its logarithmic form is the effect of the Schwarzschild horizon boundary condition on the scalar field exchanged between the shells.
The formula proves that there is an open parameter region with $a_1<1$ and $a_2<1$ in which both shells are separately stable but the pair has a growing $s$-wave mode.
In the equal-strength case, the critical normalized attraction is
$$
a_{\rm c}=\frac{1}{1+\sqrt{\chi}}.
$$

We also derive the critical static scalar cloud in closed form.
The cloud is regular at the horizon, has no node, and has a $1/r$ tail at infinity.
The analytic zero-frequency result is checked by computing the finite growth rate beyond threshold.
Finally, we express the threshold in terms of the coupling $\beta$, the shell masses, their surface pressures, and the gravitational redshift.
A large reduction of the normalized threshold does not necessarily imply scalarization of a light astrophysical shell at moderate coupling.

We first consider the Schwarzschild probe limit and obtain the threshold kernel and the scalar cloud in closed form.
We then impose the two Israel junctions on the scalar-free background and derive the static onset condition.
The rank-two determinant remains valid after the metric junctions are included.
The probe formula is recovered at leading order for small shell masses.
What remains outside the present paper is the nonlinear scalarized branch, together with a material equation of state and a radial-stability analysis for the shells.

Figure~\ref{fig:geometry} shows the arrangement.
The two delta interactions are outside the horizon and represent the
integrated stress traces of thin matter layers.

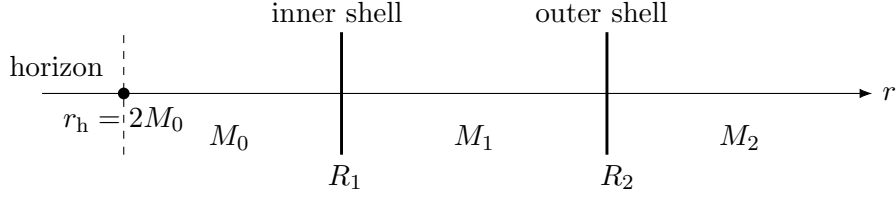
\begin{figure}[htbp]
\centering
\begin{tikzpicture}[x=0.9cm,y=0.9cm,>=Latex]

 % axis
 \draw[->] (0,0)--(12.2,0) node[right] {$r$};

 % horizon
 \fill (1.2,0) circle (2.2pt);
 \draw[dashed] (1.2,-0.9)--(1.2,0.9);
 \node[below] at (1.2,-0.05) {$r_{\rm h}=2M_0$};
 \node[above left] at (1.05,0.1) {horizon};

 % shells
 \draw[line width=1.1pt] (4.4,-0.9)--(4.4,0.9);
 \draw[line width=1.1pt] (8.3,-0.9)--(8.3,0.9);

 % shifted labels to avoid overlap with the vertical lines
 \node[below right] at (4.05,-0.9) {$R_1$};
 \node[below right] at (8.05,-0.9) {$R_2$};

 \node[above left]  at (5.45,0.9) {inner shell};
 \node[above left]  at (9.35,0.9) {outer shell};

 % region labels
 \node at (2.75,-0.65) {$M_0$};
 \node at (6.35,-0.65) {$M_1$};
 \node at (10.25,-0.65) {$M_2$};

\end{tikzpicture}

\caption{Radial schematic of a central black hole and two concentric thin shells.
In the exact scalar-free construction, the two shells divide the exterior into
three vacuum Schwarzschild regions with mass parameters $M_0$, $M_1$, and
$M_2$. The lapse normalizations are chosen so that the induced metric is
continuous across both shells. The Schwarzschild probe problem is recovered as
the leading small-shell-mass limit $M_1,M_2\to M_0=M$.}
\label{fig:geometry}
\end{figure}

\section{Model, probe limit, and thin shell reduction}

\subsection{Linear scalar equation}
\label{subsec:linear-scalar}

We use units $G=c=1$, the metric signature $(-,+,+,+)$, and $\Box=g^{\mu\nu}\nabla_\mu\nabla_\nu$.
The scalar field is taken to be dimensionless in the normalization below.
The Einstein-frame action is
\begin{equation}
S=\frac{1}{16\pi}\int \sqrt{-g}
\left\{R-2g^{\mu\nu}\nabla_\mu\Phi\nabla_\nu\Phi\right\}\dd^4x
+S_{\rm m}\left[\Psi,A^2(\Phi)g_{\mu\nu}\right].
\label{eq:action}
\end{equation}
Here $\Phi$ is the scalar field, $\Psi$ denotes the matter fields, and $A(\Phi)$ is the matter coupling function.
The matter fields are minimally coupled to the Jordan-frame metric
$$
\widetilde g_{\mu\nu}=A^2(\Phi)g_{\mu\nu}.
$$
The Einstein-frame stress tensor and its trace are defined by
$$
T_{\mu\nu}=-\frac{2}{\sqrt{-g}}\frac{\delta S_{\rm m}}{\delta g^{\mu\nu}},
\qquad
T=g^{\mu\nu}T_{\mu\nu}.
$$
Define
$$
\alpha(\Phi)=\frac{\dd\log A(\Phi)}{\dd\Phi}.
$$
Here $\log$ denotes the natural logarithm.
Variation of Eq.~\eqref{eq:action} with respect to $\Phi$ gives
\begin{equation}
\Box\Phi=-4\pi\alpha(\Phi)T.
\label{eq:scalarfull}
\end{equation}
The trace $T$ in the linear problem is evaluated on the scalar-free background.
We impose $A(0)=1$ and $\alpha(0)=0$, so the Einstein and Jordan metrics agree on that background and $\Phi=0$ is a solution.
All shell quantities used below are evaluated on this scalar-free background, where the two frame metrics coincide.
At first order, the scalar sector is autonomous.
Indeed, $\delta(\Box\Phi)=\Box_0\varphi$ because $\Phi_0=0$, the scalar stress tensor is quadratic in $\varphi$, and the variation of the matter source is
$$
-4\pi\,\delta\{\alpha(\Phi)T\}
=-4\pi\{\beta T_0\varphi+\alpha(0)\delta T\}
=-4\pi\beta T_0\varphi.
$$
Thus metric and matter perturbations are not external sources for $\varphi$ at linear order.
They may be studied separately, but they are not required for the scalar onset problem considered here.
Let
$$
\beta=\left.\frac{\dd\alpha}{\dd\Phi}\right|_{\Phi=0}.
$$
Then
$$
\alpha(\Phi)=\beta\Phi+O(\Phi^2)
$$
near $\Phi=0$.
A small perturbation $\varphi$ therefore satisfies
\begin{equation}
\left(\Box-\mu_{\rm eff}^2\right)\varphi=0,
\qquad
\mu_{\rm eff}^2=-4\pi\beta T.
\label{eq:linearscalar}
\end{equation}
A region with $\mu_{\rm eff}^2<0$ gives an attractive contribution to the radial Schr\"odinger operator and can produce a growing scalar mode.
This is the linear-stability meaning of the word ``tachyonic'' in this paper.
The zeroth-order metric is Schwarzschild,
\begin{equation}
\dd s^2=-f(r)\dd t^2+f(r)^{-1}\dd r^2+r^2\dd\Omega^2,
\qquad
f(r)=1-\frac{2M}{r},
\label{eq:schwarzschild}
\end{equation}
where $\dd\Omega^2$ is the metric of the unit two-sphere and the event horizon is at $r=2M$.

\subsection{Probe scaling and physical status of the shells}

Let $m_j$ be the total proper energy of shell $j$ and define
\begin{equation}
\epsilon_{\rm sh}=\frac{m_1+m_2}{M}.
\label{eq:epsshell}
\end{equation}
We assume $\epsilon_{\rm sh}\ll1$ and retain only the leading Schwarzschild metric,
\begin{equation}
g_{\mu\nu}=g_{\mu\nu}^{\rm Schw}+O(\epsilon_{\rm sh}).
\label{eq:metricprobe}
\end{equation}
The stress tensor of the scalar perturbation is quadratic in its amplitude and does not change the metric at linear order.
The shell metric backreaction is of order $\epsilon_{\rm sh}$, but the linear scalar equation retains the product of $\beta$ and the shell trace.
For shells with radii of order $M$ and with $|\tau_j|$ of the same order as the surface energy density, the dimensionless delta strength derived below satisfies
$$
\lambda_jM=O\left(|\beta|\frac{m_j}{M}\right).
$$
Thus $\epsilon_{\rm sh}$ controls the metric correction, whereas $|\beta|\epsilon_{\rm sh}$ controls the scalar attraction.
We use the formal double scaling
$$
\epsilon_{\rm sh}\ll1,
\qquad
|\beta|\epsilon_{\rm sh}=O(1),
$$
so the geometry is Schwarzschild at leading order while the shell term in the scalar equation remains finite.
This is an operator-level probe limit for the linear onset problem, not a weak-coupling expansion.
More precisely, it describes a family in which the shell mass ratios decrease while the scalar-matter coupling is increased so that the integrated scalar attraction remains fixed.
It is therefore distinct from sending the shell masses to zero at fixed $\beta$ in a single scalar-tensor theory.
The omitted piecewise-metric corrections change the entries of the threshold Green matrix by $O(\epsilon_{\rm sh})$.
Provided that the first critical root is simple, the corresponding critical coupling or critical curve is shifted by $O(\epsilon_{\rm sh})$.
Section~\ref{sec:exact-background} verifies this statement on an exact background with two Israel shells.
The double scaling does not imply that an arbitrarily light shell scalarizes for observationally allowed values of $\beta$.

A freely falling dust shell cannot remain at a fixed Schwarzschild radius.
The shells in this model are prescribed static layers with surface density and tangential pressure.
Any material stress used to support a shell is understood to be included in its surface stress tensor.
Alternatively, an external constraint may be regarded as part of the prescribed background and is not varied.
Cardoso et al.\ treated the single-shell problem with Israel matching and a self-consistent piecewise Schwarzschild geometry \cite{Cardoso2013PRL}.
The exact scalar-free two-shell geometry, containing three vacuum regions and two sets of junction conditions, is constructed in Sec.~\ref{sec:exact-background}.
The leading Schwarzschild probe problem gives the threshold and critical cloud in elementary closed form.
The nonlinear scalarized two-shell branch is not constructed here.

\subsection{Relation to Israel junction conditions}

A thin shell plays two different roles in the present problem.
Let $\Sigma_j$ be a non-null hypersurface separating an inner region
($-$) from an outer region ($+$).  The Israel junction conditions are

\begin{equation}
[K_{ab}]_j-\gamma_{ab}[K]_j
=-8\pi S_{ab}^{(j)},
\label{eq:israeljunction}
\end{equation}
where $\gamma_{ab}$ is the metric induced on $\Sigma_j$,
$K_{ab}^{\pm}$ are the extrinsic curvatures evaluated on the two sides
of the shell, and
$$
K^{\pm}=\gamma^{ab}K_{ab}^{\pm}.
$$
For any quantity $X$ with limiting values on the two sides of the shell,
we define its jump by
$$
[X]_j:=X_j^{+}-X_j^{-},
$$
where $+$ and $-$ denote the outer and inner sides, respectively.
The tensor $S_{ab}^{(j)}$ is the surface stress tensor of shell $j$.
Equation~\eqref{eq:israeljunction} is the distributional Einstein equation for the metric.
For two spherical shells it produces three piecewise Schwarzschild regions, with mass parameters and time normalizations related by the two junctions.
This exact scalar-free geometry is written explicitly in Sec.~\ref{sec:exact-background}.
A material equation of state is still needed to decide whether the prescribed radii are radially stable.

The scalar matching condition used in the present paper is different.
On the scalar-free branch, the trace $\tau_j=S^a{}_a$ enters Eq.~\eqref{eq:linearscalar} and produces the derivative jump in Eq.~\eqref{eq:jump}.
The Israel conditions fix the metric jump data generated by the full tensor $S_{ab}^{(j)}$, whereas the linear scalar equation selects its trace and fixes the scalar derivative jump.
Therefore, replacing the delta interactions by Israel matching alone would omit the scalar instability mechanism.
Conversely, the scalar jump condition does not replace the Israel conditions in a self-consistent spacetime construction.

The probe limit gives a precise separation between these two effects.
Equation~\eqref{eq:israeljunction} implies metric and extrinsic-curvature corrections of order $\epsilon_{\rm sh}$, which are omitted from the leading background Green function.
The scalar jump is of order $|\beta|\epsilon_{\rm sh}$ and is retained in the double scaling stated above.
Thus the use of one Schwarzschild metric on both sides of each shell is not an exact solution of the Einstein equation with a nonzero surface stress tensor; it is the leading term of the probe expansion.
At this order, inserting the $O(\epsilon_{\rm sh})$ piecewise-metric corrections into the scalar Green function would change the onset determinant only beyond the retained background order, while the delta strengths remain finite by construction.
The result below gives the scalar onset condition to leading order in this double scaling.
Section~\ref{sec:exact-background} then imposes the Israel geometry and the scalar jump simultaneously on the exact scalar-free branch.

\subsection{Distributional trace and radial jump}

Let the shells be at $2M<R_1<R_2$.
Define the outward proper radial distance $\zeta$ by $\dd\zeta=\dd r/\sqrt{f(r)}$, and let $\zeta_j=\zeta(R_j)$.
For a static spherical shell, the surface stress tensor in an orthonormal frame tangent to the shell is
$$
S^a{}_b=\diag(-\sigma_j,p_j,p_j).
$$
Its surface trace is
\begin{equation}
\tau_j=S^a{}_a=-\sigma_j+2p_j.
\label{eq:surfacetrace}
\end{equation}
The corresponding spacetime trace is the surface distribution
\begin{equation}
T=\sum_{j=1}^2\tau_j\delta_{\Sigma_j},
\label{eq:trace}
\end{equation}
where $\delta_{\Sigma_j}$ denotes the invariant delta distribution
supported on the shell $\Sigma_j$, defined by
$$
\int \delta_{\Sigma_j} F\sqrt{-g}\dd^4x
=
\int_{\Sigma_j}F\sqrt{-h_j}\dd^3y
$$
for every smooth compactly supported test function $F$.
Here $h_j$ is the induced metric on the world volume $\Sigma_j$ of
shell $j$.
In the proper radial coordinate $\zeta$, one has
$$
\delta_{\Sigma_j}=\delta(\zeta-\zeta_j),
$$
so that Eq.~\eqref{eq:trace} becomes
$$
T=\sum_{j=1}^2\tau_j\delta(\zeta-\zeta_j).
$$
This is the standard invariant normalization of a surface distribution
in the Israel formalism \cite{Israel1966}.
A trace-free shell with $\sigma_j=2p_j$ does not couple to the linear
scalar perturbation.

We introduce the tortoise coordinate $x=r_*$ by
\begin{equation}
\frac{\dd x}{\dd r}=\frac{1}{f(r)},
\qquad
x=r+2M\log\left(\frac{r}{2M}-1\right),
\label{eq:tortoise}
\end{equation}
where the additive constant has been fixed as shown.
Since
$$
\dd\zeta=\frac{\dd r}{\sqrt{f(r)}}=\sqrt{f(r)}\dd x,
$$
the delta distribution satisfies
\begin{equation}
\delta(\zeta-\zeta_j)
=\sqrt{f_j}\delta(r-R_j)
=\frac{1}{\sqrt{f_j}}\delta(x-x_j),
\qquad
f_j=f(R_j),
\qquad
x_j=r_*(R_j).
\label{eq:deltaconvert}
\end{equation}

We decompose the perturbation as
\begin{equation}
\varphi(t,r,\theta,\phi)
=\sum_{\ell,m}e^{-i\omega t}Y_{\ell m}(\theta,\phi)\frac{u_{\ell m}(x)}{r},
\label{eq:decomp}
\end{equation}
where the spherical harmonics are orthonormal on the unit two-sphere.
The radial operator does not depend on $m$, and this index will be suppressed when no confusion can arise.
Before the thin-shell substitution, Eq.~\eqref{eq:linearscalar} gives
\begin{equation}
-\frac{\dd^2u_{\ell m}}{\dd x^2}
+f(r)\left\{\frac{\ell(\ell+1)}{r^2}
+\frac{2M}{r^3}
+\mu_{\rm eff}^2(r)\right\}u_{\ell m}
=\omega^2u_{\ell m}.
\label{eq:radialbefore}
\end{equation}
Equations~\eqref{eq:linearscalar}, \eqref{eq:trace}, and \eqref{eq:deltaconvert} give
\begin{equation}
f(r)\mu_{\rm eff}^2(r)
=\sum_{j=1}^2\lambda_j\delta(x-x_j),
\qquad
\lambda_j=-4\pi\beta\tau_j\sqrt{f_j}.
\label{eq:singularmass}
\end{equation}
Consequently,
\begin{equation}
H_\ell u_{\ell m}=\omega^2u_{\ell m},
\label{eq:radialeigen}
\end{equation}
where
\begin{equation}
H_\ell=-\frac{\dd^2}{\dd x^2}+V_\ell^{\BH}(x)
+\lambda_1\delta(x-x_1)+\lambda_2\delta(x-x_2)
\label{eq:radialoperator}
\end{equation}
and
\begin{equation}
V_\ell^{\BH}(x)
=f(r)\left\{\frac{\ell(\ell+1)}{r^2}+\frac{2M}{r^3}\right\}.
\label{eq:bhpotential}
\end{equation}
The superscript $\BH$ denotes the shell-free black-hole operator.
The Schwarzschild part is the standard effective potential for a massless scalar field on Schwarzschild spacetime \cite{Cardoso2013PRD,Berti2009}.
The factor $\sqrt{f_j}$ in Eq.~\eqref{eq:singularmass} is the product of the factor $f$ in the radial mass term and the Jacobian $1/\sqrt{f_j}$ in Eq.~\eqref{eq:deltaconvert}.
It is fixed by the invariant surface trace and is not an independent potential parameter.

Integration of Eq.~\eqref{eq:radialeigen} across $x=x_j$ gives
\begin{eqnarray}
u(x_j+0)&=&u(x_j-0),
\label{eq:continuity}\\
u'(x_j+0)-u'(x_j-0)&=&\lambda_j u(x_j),
\label{eq:jump}
\end{eqnarray}
where a prime denotes differentiation with respect to $x$.
In units $G=c=1$, both $\tau_j$ and $\lambda_j$ have dimension of inverse length.
For pressureless matter, $\tau_j=-\sigma_j$.
The standard scalarizing sign $\beta<0$ then gives $\lambda_j<0$, so a dust shell is attractive.

\subsection{Self-adjoint radial operator and instability}

The background potential is nonnegative and tends to zero at both ends of the tortoise line.
The quadratic form of $H_\ell$ is
\begin{equation}
q_\ell[u]=\int_{\R}\left\{|u'(x)|^2+V_\ell^{\BH}(x)|u(x)|^2\right\}\dd x
+\lambda_1|u(x_1)|^2+\lambda_2|u(x_2)|^2,
\label{eq:quadraticform}
\end{equation}
with form domain $H^1(\R)$.
For every $\varepsilon>0$, point evaluation satisfies an estimate of the form
$$
|u(x_j)|^2
\leq \varepsilon\int_{\R}|u'(x)|^2\dd x
+C_{\varepsilon,j}\int_{\R}|u(x)|^2\dd x.
$$
The shell terms are therefore form bounded with relative bound zero.
It follows that $q_\ell$ is closed and lower bounded and defines a self-adjoint operator.
Multiplication by $V_\ell^{\BH}$ is relatively compact with respect to $-\dd^2/\dd x^2$, and the two point interactions give a finite-rank resolvent difference.
Hence the essential spectrum is $[0,\infty)$, and the negative spectrum consists of isolated eigenvalues of finite multiplicity.

With the time dependence in Eq.~\eqref{eq:decomp}, a negative eigenvalue
$$
\omega^2=-\kappa^2,
\qquad
\kappa>0,
$$
produces a solution proportional to $e^{\kappa t}$.
The corresponding radial function is square integrable and satisfies
\begin{eqnarray}
u(x)&\sim&e^{\kappa x},
\qquad x\longrightarrow-\infty,
\label{eq:leftdecay}\\
u(x)&\sim&e^{-\kappa x},
\qquad x\longrightarrow+\infty.
\label{eq:rightdecay}
\end{eqnarray}
Near the future horizon, the full mode is proportional to $e^{\kappa(t+x)}$, which is regular in the ingoing coordinate $v=t+x$.
Thus a negative eigenvalue of $H_\ell$ is exactly a growing scalar mode in the present linear problem.

\section{Boundary matrix and unstable spectrum}

\subsection{Schwarzschild Green function}

Let
$$
H_\ell^{\BH}=-\frac{\dd^2}{\dd x^2}+V_\ell^{\BH}(x)
$$
be the shell-free radial operator.
For $\kappa>0$, let $G_\ell(x,y;-\kappa^2)$ be the kernel of $(H_\ell^{\BH}+\kappa^2)^{-1}$.
Let $p_{\ell,\kappa}$ and $q_{\ell,\kappa}$ be the positive solutions of
$$
\left(H_\ell^{\BH}+\kappa^2\right)w=0
$$
normalized by
$$
\lim_{x\to-\infty}e^{-\kappa x}p_{\ell,\kappa}(x)=1,
\qquad
\lim_{x\to+\infty}e^{\kappa x}q_{\ell,\kappa}(x)=1.
$$
Their Wronskian
$$
W_{\ell,\kappa}
=p_{\ell,\kappa}q'_{\ell,\kappa}
-p'_{\ell,\kappa}q_{\ell,\kappa}
$$
is independent of $x$.
The Green kernel is
\begin{equation}
G_\ell(x,y;-\kappa^2)
=-\frac{p_{\ell,\kappa}(x_{<})q_{\ell,\kappa}(x_{>})}{W_{\ell,\kappa}},
\qquad
x_{<}=\min\{x,y\},
\qquad
x_{>}=\max\{x,y\}.
\label{eq:greennegative}
\end{equation}
The one-dimensional Schr\"odinger semigroup for $H_\ell^{\BH}$ is positivity improving, so its resolvent kernel is strictly positive.

Let $I_n$ denote the $n\times n$ identity matrix.
We define
\begin{equation}
\mathsf M_\ell(\kappa)
=\left(
\begin{array}{cc}
G_\ell(x_1,x_1;-\kappa^2)&G_\ell(x_1,x_2;-\kappa^2)\\
G_\ell(x_2,x_1;-\kappa^2)&G_\ell(x_2,x_2;-\kappa^2)
\end{array}
\right)
\label{eq:Mmatrix}
\end{equation}
and
\begin{equation}
\Lambda=\left(
\begin{array}{cc}
\lambda_1&0\\
0&\lambda_2
\end{array}
\right),
\qquad
K_\ell(\kappa)=I_2+\mathsf M_\ell(\kappa)\Lambda.
\label{eq:Kmatrix}
\end{equation}
For $z\in\C\setminus[0,\infty)$, we also write
\begin{equation}
{\mathcal K}_\ell(z)
=I_2+{\mathcal M}_\ell(z)\Lambda,
\qquad
[{\mathcal M}_\ell(z)]_{ij}=G_\ell(x_i,x_j;z),
\label{eq:Kz}
\end{equation}
where $G_\ell(x,y;z)$ is the kernel of $(H_\ell^{\BH}-z)^{-1}$.
Thus $K_\ell(\kappa)={\mathcal K}_\ell(-\kappa^2)$.
This matrix is the Schwarzschild counterpart of the reduced boundary matrices for concentric delta shells in flat space \cite{Kaminaga2026}.

\begin{theorem}[Boundary determinant]
\label{thm:boundary-determinant}
Let $\kappa>0$.
Then
\begin{equation}
\dim\ker(H_\ell+\kappa^2)=\dim\ker K_\ell(\kappa).
\label{eq:kernelmultiplicity}
\end{equation}
In particular, $-\kappa^2$ is an eigenvalue of $H_\ell$ if and only if
\begin{equation}
D_\ell(\kappa):=\det K_\ell(\kappa)=0.
\label{eq:detcondition}
\end{equation}
If $K_\ell(\kappa)$ is invertible, the resolvent kernel $G_{\ell,\Lambda}$ of $H_\ell$ is
\begin{eqnarray}
G_{\ell,\Lambda}(x,y;-\kappa^2)
&=&G_\ell(x,y;-\kappa^2)
\nonumber\\
&&-\sum_{i,j=1}^2
G_\ell(x,x_i;-\kappa^2)\lambda_i
\left[K_\ell(\kappa)^{-1}\right]_{ij}
G_\ell(x_j,y;-\kappa^2).
\label{eq:resolventformula}
\end{eqnarray}
\end{theorem}

\begin{proof}
Let $f\in L^2(\R)$ and put $u_0=(H_\ell^{\BH}+\kappa^2)^{-1}f$.
A solution of $(H_\ell+\kappa^2)u=f$ must satisfy
$$
u(x)=u_0(x)-\sum_{j=1}^2\lambda_jc_jG_\ell(x,x_j;-\kappa^2),
$$
where $c_j=u(x_j)$.
The Green kernel is continuous, and its derivative with respect to the first variable has jump $-1$ at the source point.
The minus sign in the representation of $u$ therefore gives the jump in Eq.~\eqref{eq:jump}.
Taking the values at the two shell points gives
$$
K_\ell(\kappa)
\left(
\begin{array}{c}
c_1\\
c_2
\end{array}
\right)
=
\left(
\begin{array}{c}
u_0(x_1)\\
u_0(x_2)
\end{array}
\right).
$$
If $K_\ell(\kappa)$ is invertible, substitution of this solution for the shell values gives Eq.~\eqref{eq:resolventformula}.

If $c$ is a nonzero vector in $\ker K_\ell(\kappa)$, then
$$
u(x)=-\sum_{j=1}^2\lambda_jc_jG_\ell(x,x_j;-\kappa^2)
$$
is square integrable, satisfies the jump conditions, and has $u(x_i)=c_i$.
It is therefore a nonzero element of $\ker(H_\ell+\kappa^2)$.
Conversely, the two shell values of any element of $\ker(H_\ell+\kappa^2)$ form a vector in $\ker K_\ell(\kappa)$.
These two maps are inverse to each other, which proves Eq.~\eqref{eq:kernelmultiplicity} and the determinant condition.
\end{proof}

For two shells, the determinant is
\begin{equation}
D_\ell(\kappa)
=\left(1+\lambda_1G_{11}\right)
\left(1+\lambda_2G_{22}\right)
-\lambda_1\lambda_2G_{12}^2,
\label{eq:generalD}
\end{equation}
where $G_{ij}=G_\ell(x_i,x_j;-\kappa^2)$.
The last term is the interaction term that allows collective scalarization.

\subsection{The first unstable sector and the number of modes}

\begin{theorem}[Angular momentum ordering]
Assume that $\lambda_1$ and $\lambda_2$ are real.
If $H_\ell$ has a negative eigenvalue for some $\ell\geq1$, then $H_0$ has a negative eigenvalue.
For each fixed $\ell$, the operator $H_\ell$ has at most two negative radial eigenvalues, counted with multiplicity.
\end{theorem}

\begin{proof}
The background potentials satisfy
$$
V_\ell^{\BH}(x)=V_0^{\BH}(x)
+f(r)\frac{\ell(\ell+1)}{r^2}.
$$
Hence $q_\ell[u]\geq q_0[u]$ for every $u\in H^1(\R)$.
If $q_\ell$ is negative on a nonzero vector, then $q_0$ is negative on the same vector.
The min-max principle therefore proves the first assertion.

Suppose that one angular sector had at least three negative radial eigenvalues.
The corresponding negative spectral subspace would contain a three-dimensional subspace on which $q_\ell[u]<0$ for every nonzero $u$.
The linear map
$$
u\longmapsto\left(u(x_1),u(x_2)\right)
$$
from this subspace to $\C^2$ has a nonzero kernel vector.
For that vector, both shell terms in Eq.~\eqref{eq:quadraticform} vanish, while the background part is nonnegative.
This contradicts $q_\ell[u]<0$.
Thus the number of negative radial eigenvalues is at most two.
\end{proof}

\begin{corollary}
Suppose that the magnitudes of the attractive shell strengths are increased monotonically from zero along a fixed one-parameter path.
The first loss of nonnegativity occurs in the $s$-wave sector, possibly at the same parameter value as another sector.
Within each fixed angular momentum sector, there are at most two radial growth rates.
\end{corollary}

The angular degeneracy of a radial eigenvalue in the $\ell$ sector is $2\ell+1$.
The theorem does not exclude a higher-$\ell$ instability after the $s$-wave sector has already become unstable.

\section{Exact onset on a two-shell Israel background}
\label{sec:exact-background}

\subsection{Static scalar resistance and the finite-rank onset condition}

The rank-two onset law is not restricted to a single Schwarzschild metric.
Consider an asymptotically flat static spherical black-hole exterior of the form
\begin{equation}
\dd s^2=-h(r)\dd t^2+\frac{\dd r^2}{f(r)}+r^2\dd\Omega^2,
\qquad r_{\rm h}<r<\infty,
\label{eq:general-static-metric}
\end{equation}
where $h$ is continuous, $h(r)>0$ and $f(r)>0$ outside a regular horizon, and the coefficients may be piecewise smooth across static shells.
We assume $h(r),f(r)\to1$ as $r\to\infty$ and that the usual horizon regularity condition excludes the nonconstant static solution there.
In addition, throughout this subsection the scalar perturbation is assumed to have no smooth bulk effective-mass term in the open regions between the shells; in the theory of Sec.~\ref{subsec:linear-scalar}, this means $\beta T_0=0$ away from the shell surfaces.
This assumption is satisfied by the piecewise-vacuum Schwarzschild construction below.
If a nonzero smooth bulk trace is present, the same finite-rank Birman--Schwinger strategy applies after replacing the elementary kernel below by the zero-energy Green function of the corresponding bulk radial operator, but the simple resistance formula need not hold.
Put
\begin{equation}
P(r)=r^2\sqrt{h(r)f(r)},
\qquad
W(r)=\frac{r^2}{\sqrt{h(r)f(r)}},
\qquad
{\cal S}(R)=\int_R^\infty\frac{\dd r}{P(r)}.
\label{eq:scalar-resistance}
\end{equation}
The function ${\cal S}$ is positive and decreasing.
The quantity ${\cal S}(R)$ relates the static scalar-field difference
between radius $R$ and infinity to the conserved radial scalar flux.
In analogy with a one-dimensional resistance, we call it the
static scalar resistance.

Let a shell at $R_j$ have invariant surface trace $\tau_j$.
Integrating the scalar equation in Gaussian normal distance across the shell gives
\begin{equation}
\left[P(r)\varphi'(r)\right]_{R_j-0}^{R_j+0}
=-s_j\varphi(R_j),
\qquad
s_j=4\pi\beta\tau_j R_j^2\sqrt{h(R_j)}.
\label{eq:general-flux-jump}
\end{equation}
The lapse $h(R_j)$ is single valued because the induced metric is continuous.
The attractive sign is $s_j>0$.
For a mode $e^{-i\omega t}Y_{\ell m}\varphi_\ell(r)$, the radial equation can be written as the generalized Sturm--Liouville problem
\begin{equation}
-\frac{\dd}{\dd r}\left(P(r)\frac{\dd\varphi_\ell}{\dd r}\right)
+\ell(\ell+1)\sqrt{\frac{h(r)}{f(r)}}\,\varphi_\ell
-\sum_j s_j\delta(r-R_j)\varphi_\ell
=\omega^2W(r)\varphi_\ell.
\label{eq:general-radial-equation}
\end{equation}
The weight $W$ is positive.
For a static $s$-wave, the vacuum equation is $(P\varphi')'=0$.
The horizon-regular solution is the constant function, while the solution decaying at infinity is ${\cal S}(r)$.
Consequently, the zero-frequency Green kernel is
\begin{equation}
{\cal G}_0(r,r')={\cal S}\bigl(\max\{r,r'\}\bigr).
\label{eq:general-static-green}
\end{equation}

\begin{theorem}[Shell onset condition]
\label{thm:universal-onset}
Assume the bulk scalar effective-mass term vanishes away from the shells, as stated above, and let $N$ attractive shells be placed at
$$
r_{\rm h}<R_1<\cdots<R_N.
$$
Define
\begin{equation}
[\mathsf M_N]_{ij}
={\cal S}\bigl(\max\{R_i,R_j\}\bigr),
\qquad
\mathsf S_N=\diag(s_1,\ldots,s_N).
\label{eq:general-threshold-matrix}
\end{equation}
A nonzero static scalar cloud, regular at the horizon and decaying at infinity, exists if and only if
\begin{equation}
\det\left(I_N-\mathsf M_N\mathsf S_N\right)=0.
\label{eq:general-threshold-determinant}
\end{equation}
Under a monotone common scaling of the attractive strengths from zero, the first scalar instability occurs in the $s$-wave sector when
\begin{equation}
\lambda_{\max}\left(
\mathsf S_N^{1/2}\mathsf M_N\mathsf S_N^{1/2}
\right)=1.
\label{eq:general-lambda-max}
\end{equation}
For two shells, set
\begin{equation}
{\cal S}_j={\cal S}(R_j),
\qquad
a_j=s_j{\cal S}_j,
\qquad
\chi=\frac{{\cal S}_2}{{\cal S}_1}.
\label{eq:general-normalized-parameters}
\end{equation}
Here $a_j=1$ is the one-shell scalar threshold obtained by retaining
only the $j$th scalar interaction on the same fixed background geometry.
Then $0<\chi<1$ and the threshold condition is
\begin{equation}
1-a_1-a_2+(1-\chi)a_1a_2=0.
\label{eq:general-two-shell-curve}
\end{equation}
If $0\leq a_1,a_2<1$, the sign of the left side distinguishes the stable side, the simple threshold, and the side with exactly one unstable $s$-wave radial mode, as in Theorem~\ref{thm:collective-threshold} below.
\end{theorem}

\begin{proof}
A static solution generated by its shell values $c_j=\varphi(R_j)$ is
$$
\varphi(r)=\sum_{j=1}^N s_jc_j
{\cal G}_0(r,R_j).
$$
Evaluation at the shell radii gives
$(I_N-\mathsf M_N\mathsf S_N)c=0$, proving
Eq.~\eqref{eq:general-threshold-determinant}.
The time-dependent radial equation is a generalized Sturm--Liouville problem with positive weight and quadratic form
$$
Q_0[\varphi]
=\int_{r_{\rm h}}^\infty P(r)|\varphi'(r)|^2\dd r
-\sum_{j=1}^Ns_j|\varphi(R_j)|^2.
$$
The shell-free form is nonnegative.
For $\omega^2=-\kappa^2$, let ${\cal G}_\kappa$ be the positive Green kernel of
$$
-\frac{\dd}{\dd r}P(r)\frac{\dd}{\dd r}+\kappa^2W(r)
$$
with horizon regularity and decay at infinity, and let
$[\mathsf M_N(\kappa)]_{ij}={\cal G}_\kappa(R_i,R_j)$.
The Birman--Schwinger matrix is
$\mathsf S_N^{1/2}\mathsf M_N(\kappa)\mathsf S_N^{1/2}$.
The resolvent identity gives, for every nonzero $z\in\C^N$ and $\kappa>0$,
$$
\frac{\dd}{\dd\kappa}
\sum_{i,j=1}^N\overline{z_i}
[\mathsf S_N^{1/2}\mathsf M_N(\kappa)\mathsf S_N^{1/2}]_{ij}z_j
=-2\kappa\int_{r_{\rm h}}^\infty
W(r)\left|\sum_{j=1}^N\sqrt{s_j}z_j
{\cal G}_\kappa(r,R_j)\right|^2\dd r<0.
$$
Moreover, $\mathsf M_N(0)=\mathsf M_N$ and
$\mathsf M_N(\kappa)\to0$ as $\kappa\to\infty$.
The Birman--Schwinger principle therefore gives Eq.~\eqref{eq:general-lambda-max} and shows that crossing it creates a growing mode.
For $\ell\geq1$, the quadratic form contains in addition the strictly nonnegative angular term
$$
\ell(\ell+1)
\int_{r_{\rm h}}^\infty
\sqrt{\frac{h(r)}{f(r)}}\,|\varphi(r)|^2\dd r,
$$
so the first loss of nonnegativity is an $s$-wave.
For $N=2$,
$$
\mathsf M_2=
\left(
\begin{array}{cc}
{\cal S}_1&{\cal S}_2\\
{\cal S}_2&{\cal S}_2
\end{array}
\right),
$$
and direct evaluation of the determinant gives Eq.~\eqref{eq:general-two-shell-curve}.
The classification for $a_1,a_2<1$ follows from Sylvester's criterion and the monotonic Birman--Schwinger eigenvalues.
\end{proof}

At threshold, a corresponding cloud is given in the exact quadrature form
\begin{equation}
\varphi_{\rm c}(r)
=\sum_{j=1}^Ns_jc_j
{\cal S}\bigl(\max\{r,R_j\}\bigr),
\label{eq:general-critical-cloud}
\end{equation}
where $c$ is a positive Perron--Frobenius eigenvector at the first onset.
Thus the first cloud is node free.

\subsection{Three Schwarzschild regions joined by two Israel shells}

We now impose the metric junctions exactly on the scalar-free branch.
Let
$$
M_0<M_1<M_2,
\qquad
R_1>2M_1,
\qquad
R_2>2M_2,
$$
and define
$$
f_k(r)=1-\frac{2M_k}{r},
\qquad k=0,1,2.
$$
With a time coordinate normalized at infinity, the metric is
\begin{equation}
h(r)=
\left\{
\begin{array}{ll}
C_0f_0(r),&2M_0<r<R_1,\\
C_1f_1(r),&R_1<r<R_2,\\
f_2(r),&r>R_2,
\end{array}
\right.
\qquad
f(r)=f_k(r)\quad\hbox{in region }k,
\label{eq:three-region-metric}
\end{equation}
where continuity of the induced metric gives
\begin{equation}
C_1=\frac{f_2(R_2)}{f_1(R_2)},
\qquad
C_0=C_1\frac{f_1(R_1)}{f_0(R_1)}.
\label{eq:lapse-normalizations}
\end{equation}
The horizon is at $r=2M_0$.
The Israel equations at shell $j$ give exactly
\begin{eqnarray}
\sigma_j
&=&\frac{\sqrt{f_{j-1}(R_j)}-\sqrt{f_j(R_j)}}{4\pi R_j},
\label{eq:exact-two-shell-sigma}\\
p_j
&=&\frac{1}{8\pi R_j}
\left\{
\frac{1-M_j/R_j}{\sqrt{f_j(R_j)}}
-\frac{1-M_{j-1}/R_j}{\sqrt{f_{j-1}(R_j)}}
\right\},
\label{eq:exact-two-shell-pressure}\\
\tau_j&=&-\sigma_j+2p_j.
\label{eq:exact-two-shell-trace}
\end{eqnarray}
Using $1-M_k/R=(1+f_k)/2$, the exact trace can also be written in the factorized form
\begin{equation}
\tau_j=
\frac{\sqrt{f_j(R_j)}-\sqrt{f_{j-1}(R_j)}}{8\pi R_j}
\left\{3-\frac{1}{\sqrt{f_{j-1}(R_j)f_j(R_j)}}\right\}.
\label{eq:exact-trace-factorization}
\end{equation}
For a nontrivial positive-mass shell, $M_j>M_{j-1}$, the prefactor in Eq.~\eqref{eq:exact-trace-factorization} is nonzero, and hence
\begin{equation}
\tau_j=0
\quad\Longleftrightarrow\quad
\sqrt{f_{j-1}(R_j)f_j(R_j)}=\frac{1}{3}.
\label{eq:exact-trace-zero}
\end{equation}
Equation~\eqref{eq:exact-trace-zero} is the finite-shell-mass continuation of the $R=3M$ transition obtained in the light-shell limit below.
For unequal adjacent masses, the exact transition radius is not, in general, the photon-sphere radius of either neighboring Schwarzschild region.
The lapse constants do not appear in Eqs.~\eqref{eq:exact-two-shell-sigma} and
\eqref{eq:exact-two-shell-pressure}, because they cancel in the ratio of the temporal metric derivative to the temporal metric itself.
They do enter the scalar strengths through Eq.~\eqref{eq:general-flux-jump}.

For this geometry, the two scalar resistances are elementary:
\begin{eqnarray}
{\cal S}_2
&=&-\frac{\log f_2(R_2)}{2M_2},
\label{eq:exact-resistance-2}\\
{\cal S}_1
&=&{\cal S}_2
+\frac{1}{2M_1\sqrt{C_1}}
\log\left\{\frac{f_1(R_2)}{f_1(R_1)}\right\}.
\label{eq:exact-resistance-1}
\end{eqnarray}
The exact attractive strengths are
\begin{equation}
s_j=4\pi\beta\tau_jR_j^2\sqrt{h(R_j)}.
\label{eq:exact-shell-strength}
\end{equation}
Equations~\eqref{eq:general-two-shell-curve} and
\eqref{eq:general-normalized-parameters}, with
Eqs.~\eqref{eq:exact-resistance-2}--\eqref{eq:exact-shell-strength}, are therefore the exact linear scalar onset condition on a scalar-free spacetime with two Israel shells.
The background metric and surface stresses satisfy the distributional Einstein equation exactly.
No probe approximation is used in this subsection.
The remaining idealization is that the surface equation of state and radial mechanical stability are not specified, and the nonlinear scalarized branch is not constructed.

\subsection{Recovery of the Schwarzschild probe limit}

Let $M_1-M_0$ and $M_2-M_1$ be of order $\epsilon M_0$.
Then
$$
C_0=1+O(\epsilon),
\qquad
C_1=1+O(\epsilon),
$$
and, for fixed radii,
\begin{equation}
{\cal S}(R_j)
=-\frac{1}{2M_0}\log\left(1-\frac{2M_0}{R_j}\right)
+O\left(\frac{\epsilon}{M_0}\right).
\label{eq:resistance-probe-limit}
\end{equation}
The exact traces and lapse factors likewise reduce to the leading expressions in Sec.~\ref{subsec:physical-parameters}.
Thus the entries of the exact threshold matrix differ from their probe values by $O(\epsilon)$.
For a simple first crossing, standard perturbation of the largest matrix eigenvalue gives an $O(\epsilon)$ shift of the critical coupling.

Table~\ref{tab:exact-probe-comparison} gives a direct check.
We take $M_0=1$, $R_1=4$, $R_2=8$, and equal gravitational mass increments
$M_1-M_0=M_2-M_1=\delta M$.
For the probe calculation, put
$$
b=\frac{|\beta|\delta M}{M_0},
\qquad
d_j=M_0\left\{1-\frac{M_0}{R_jf_0(R_j)}\right\}
{\cal S}_0(R_j),
$$
where ${\cal S}_0$ is the Schwarzschild resistance with mass $M_0$.
Then $a_j=bd_j$ and the smaller positive root of
$$
1-b(d_1+d_2)+(1-\chi_0)b^2d_1d_2=0
$$
gives $b_{\rm c}^{\rm probe}=4.11199$.
The exact column uses Eqs.~\eqref{eq:exact-two-shell-sigma}--\eqref{eq:exact-shell-strength} in the same two-shell determinant and solves for the common coupling $|\beta|$.

\begin{table}[H]
\centering
\caption{Comparison between the exact two-Israel-shell onset and the Schwarzschild probe approximation for $M_0=1$, $R_1=4$, $R_2=8$, and equal mass increments $\delta M$.
The displayed coupling is $|\beta_{\rm c}|\delta M/M_0$.
The relative shift is $(|\beta_{\rm c}^{\rm exact}|-|\beta_{\rm c}^{\rm probe}|)/|\beta_{\rm c}^{\rm probe}|$.}
\label{tab:exact-probe-comparison}
\begin{tabular}{cccc}
\toprule
$2\delta M/M_0$ & probe & exact & relative shift\\
\midrule
$0.005$ & $4.11199$ & $4.11849$ & $0.158\%$\\
$0.010$ & $4.11199$ & $4.12505$ & $0.318\%$\\
$0.020$ & $4.11199$ & $4.13834$ & $0.641\%$\\
$0.040$ & $4.11199$ & $4.16559$ & $1.304\%$\\
\bottomrule
\end{tabular}
\end{table}

The discrepancy is linear in the shell mass ratio, as expected from the probe expansion.

\section{Closed-form Schwarzschild onset and critical scalar cloud}

\subsection{Threshold Green function}

At $\omega=0$ and $\ell=0$, the shell-free radial equation has the two solutions
\begin{equation}
y_{\rm L}(r)=r,
\qquad
y_{\rm R}(r)=-\frac{r}{2M}\log f(r).
\label{eq:zeromodes}
\end{equation}
The scalar field $y_{\rm L}/r=1$ is regular at the future horizon.
The function $y_{\rm R}$ tends to $1$ at infinity, so $y_{\rm R}/r$ has the required $1/r$ falloff.
Their Wronskian with respect to $x$ is
\begin{equation}
W_x[y_{\rm L},y_{\rm R}]=-1.
\label{eq:zerowronskian}
\end{equation}

\begin{lemma}[Threshold kernel]
\label{lem:threshold-kernel}
For fixed $x$ and $y$, the negative-energy Green kernel has the limit
\begin{equation}
G_0^{\BH}(x,y;0)
:=\lim_{\kappa\downarrow0}G_0(x,y;-\kappa^2)
=y_{\rm L}(r_{<})y_{\rm R}(r_{>}).
\label{eq:zerogreen}
\end{equation}
Here $r_{<}=\min\{r(x),r(y)\}$ and $r_{>}=\max\{r(x),r(y)\}$.
The convergence is locally uniform in $x$ and $y$.
\end{lemma}

\begin{proof}
The potential $V_0^{\BH}$ decays exponentially as $x\to-\infty$ and as $O(x^{-3})$ as $x\to+\infty$.
In particular,
$$
\int_{\R}(1+|x|)V_0^{\BH}(x)\dd x<\infty.
$$
The normalized left and right solutions satisfy the Volterra equations
\begin{eqnarray}
p_{0,\kappa}(x)
&=&e^{\kappa x}
+\int_{-\infty}^{x}
\frac{\sinh\{\kappa(x-s)\}}{\kappa}
V_0^{\BH}(s)p_{0,\kappa}(s)\dd s,
\nonumber\\
q_{0,\kappa}(x)
&=&e^{-\kappa x}
+\int_x^{\infty}
\frac{\sinh\{\kappa(s-x)\}}{\kappa}
V_0^{\BH}(s)q_{0,\kappa}(s)\dd s.
\label{eq:volterra-zero}
\end{eqnarray}
The integrability above and dominated convergence show that the solutions and their first derivatives converge locally uniformly as $\kappa\downarrow0$.
The limiting left solution is normalized to tend to $1$ at the horizon and is therefore $y_{\rm L}/(2M)$.
The limiting right solution is normalized to tend to $1$ at infinity and is therefore $y_{\rm R}$.
Their Wronskian is $-1/(2M)$ by Eq.~\eqref{eq:zerowronskian}.
Substitution into Eq.~\eqref{eq:greennegative} gives Eq.~\eqref{eq:zerogreen}.
\end{proof}

The kernel in Eq.~\eqref{eq:zerogreen} imposes regularity at the horizon and a $1/r$ falloff for the scalar field at infinity.
It is not the kernel of a bounded inverse on $L^2(\R)$ because zero is the lower edge of the continuous spectrum.
In what follows, a threshold solution means a nonzero zero-frequency solution that is regular at the future horizon and whose scalar field has the decaying multipole behavior $r^{-\ell-1}$ at infinity.
Set
\begin{equation}
L_j=-\log f(R_j)>0,
\qquad j=1,2.
\label{eq:Lj}
\end{equation}
Since $R_1<R_2$, Eq.~\eqref{eq:zerogreen} gives
\begin{eqnarray}
G_{11}^{(0)}&=&\frac{R_1^2L_1}{2M},
\label{eq:G11zero}\\
G_{22}^{(0)}&=&\frac{R_2^2L_2}{2M},
\label{eq:G22zero}\\
G_{12}^{(0)}&=&G_{21}^{(0)}
=\frac{R_1R_2L_2}{2M}.
\label{eq:G12zero}
\end{eqnarray}
Let $\mathsf M_0(0)$ be the $2\times2$ matrix with these entries.
The geometric overlap parameter is
\begin{equation}
\chi
=\frac{\left(G_{12}^{(0)}\right)^2}{G_{11}^{(0)}G_{22}^{(0)}}
=\frac{L_2}{L_1}
=\frac{\log f(R_2)}{\log f(R_1)}.
\label{eq:chi}
\end{equation}
Because $0<f(R_1)<f(R_2)<1$, one has $0<\chi<1$.
The limit $\chi\to1$ corresponds to neighboring shells, whereas $\chi\to0$ corresponds to weak static overlap.
For fixed $R_2$, the near-horizon limit $R_1\downarrow2M$ gives $L_1\to\infty$ and therefore $\chi\to0$.
Thus an inner layer moved toward the horizon becomes weakly coupled, in normalized static-response units, to a fixed outer layer.
This suppression is a specifically black-hole effect associated with the horizon boundary condition and the large tortoise separation.

We next examine the flat-space limit.
For $M/R_j\to0$,
$$
L_j=\frac{2M}{R_j}+O\left(\frac{M^2}{R_j^2}\right),
$$
and hence
\begin{equation}
G_{11}^{(0)}\longrightarrow R_1,
\qquad
G_{22}^{(0)}\longrightarrow R_2,
\qquad
G_{12}^{(0)}\longrightarrow R_1
\label{eq:flatgreen}
\end{equation}
and
\begin{equation}
\chi\longrightarrow\frac{R_1}{R_2}.
\label{eq:flatchi}
\end{equation}
These are the zero-energy radial Green matrix entries for two concentric delta shells in flat three-dimensional space.
The critical attraction of one shell is
\begin{equation}
g_{j,{\rm c}}=\frac{2M}{R_j^2L_j},
\label{eq:singlecritical}
\end{equation}
which tends to $1/R_j$ in the flat limit.
Thus the Schwarzschild threshold reduces continuously to the single-shell condition $\lambda_j=-1/R_j$ used in the flat delta-shell problem \cite{Kaminaga2026}.

Assume from now on that both shell interactions are attractive and write
\begin{equation}
g_j=-\lambda_j>0.
\label{eq:gj}
\end{equation}
The normalized single-shell strengths are
\begin{equation}
a_j=g_jG_{jj}^{(0)}.
\label{eq:aj}
\end{equation}
A shell at $R_j$ reaches its isolated static threshold at $a_j=1$.
The threshold Birman--Schwinger matrix is
\begin{equation}
B_0=
\left(
\begin{array}{cc}
a_1&\sqrt{a_1a_2\chi}\\
\sqrt{a_1a_2\chi}&a_2
\end{array}
\right).
\label{eq:Bzero}
\end{equation}
Its off-diagonal entry measures the static scalar exchange between the shells.

\begin{theorem}[Collective scalarization curve in Schwarzschild]
\label{thm:collective-threshold}
Assume $0\leq a_1<1$ and $0\leq a_2<1$, so each shell is below its one-shell threshold.
Define
\begin{equation}
\Delta_0
=1-a_1-a_2+(1-\chi)a_1a_2.
\label{eq:Delta0}
\end{equation}
If $\Delta_0>0$, the complete scalar perturbation problem has no growing mode in any angular sector.
If $\Delta_0=0$, the scalar perturbation problem has no growing mode and has a unique $s$-wave threshold solution up to normalization.
If $\Delta_0<0$, the $s$-wave sector has exactly one growing radial mode.
The critical curve is
\begin{equation}
a_2=\frac{1-a_1}{1-(1-\chi)a_1}.
\label{eq:criticalcurve}
\end{equation}
\end{theorem}

\begin{proof}
For $\kappa>0$, define the symmetric Birman--Schwinger matrix
$$
[B(\kappa)]_{ij}
=\sqrt{g_i}\,G_0(x_i,x_j;-\kappa^2)\sqrt{g_j}.
$$
The boundary equation $K_0(\kappa)c=0$ is equivalent, after multiplication by the positive diagonal matrix $\Gamma^{1/2}=\diag(\sqrt{g_1},\sqrt{g_2})$, to
$$
\left(I_2-B(\kappa)\right)\Gamma^{1/2}c=0.
$$
Thus $-\kappa^2$ is an $s$-wave eigenvalue, with the same multiplicity, if and only if $1$ is an eigenvalue of $B(\kappa)$.

The matrix $B(\kappa)$ is continuous for $\kappa\geq0$, with $B(0)=B_0$, and tends to zero as $\kappa\to\infty$.
For every nonzero $z\in\C^2$ and $\kappa>0$, the resolvent identity gives
\begin{eqnarray}
\frac{\dd}{\dd\kappa}
\sum_{i,j=1}^2\overline{z_i}[B(\kappa)]_{ij}z_j
&=&-2\kappa\int_{\R}
\left|
\sum_{j=1}^2\sqrt{g_j}z_j
G_0(y,x_j;-\kappa^2)
\right|^2\dd y
\nonumber\\
&<&0.
\label{eq:Bmonotone}
\end{eqnarray}
The inequality is strict because a nonzero linear combination of delta sources cannot be annihilated by the resolvent.
Hence $B(\kappa)$ is strictly decreasing in the matrix order.
Let $b_1(\kappa)\geq b_2(\kappa)$ be its ordered eigenvalues.
They are continuous, strictly decreasing, and tend to zero as $\kappa\to\infty$.
By the boundary determinant theorem, each solution of $b_j(\kappa)=1$ gives the negative eigenvalue $-\kappa^2$ with the same multiplicity.
The threshold Green kernel in Lemma~\ref{lem:threshold-kernel} and the same shell-value argument identify $\ker(I_2-B_0)$ with the space of $s$-wave threshold solutions.
Therefore, the number of negative $s$-wave eigenvalues equals the number of eigenvalues of $B_0$ larger than $1$.
An eigenvalue of $B_0$ equal to $1$ gives a threshold solution of the same multiplicity.

Because $a_1<1$, the first leading principal minor of $I_2-B_0$ is positive.
Its determinant is
$$
\det(I_2-B_0)=\Delta_0.
$$
If $\Delta_0>0$, Sylvester's criterion gives $I_2-B_0>0$, so the $s$-wave sector has no negative eigenvalue.
Angular momentum ordering then excludes a negative eigenvalue in every higher sector.

If $\Delta_0=0$, the largest eigenvalue of $B_0$ is $1$ and is simple, while the other eigenvalue is smaller than $1$.
This gives one threshold solution and no negative $s$-wave eigenvalue.
A negative mode for $\ell\geq1$ is excluded by angular momentum ordering.
Suppose instead that a nonzero higher-$\ell$ threshold solution exists.
Multiply its zero-energy equation by its complex conjugate, integrate piecewise across the shells, and use the jump conditions.
After the endpoints are moved to the two scattering ends, the boundary terms vanish because the solution is regular at the horizon and has the decaying static behavior at infinity.
The resulting improper quadratic-form identity is $q_\ell[u]=0$.
To justify the comparison with the $s$-wave form, choose $\eta_L\in C^\infty(\R)$ such that $\eta_L=0$ on $(-\infty,-2L)$, $\eta_L=1$ on $(-L,\infty)$, and $|\eta_L'|\leq C/L$.
The threshold solution is bounded at the horizon, and hence
$$
\int_{-2L}^{-L}|\eta_L'(x)u(x)|^2\dd x=O(L^{-1}).
$$
The already established absence of negative $s$-wave spectrum gives $q_0[\eta_Lu]\geq0$.
Letting $L\to\infty$ yields the improper inequality $q_0[u]\geq0$.
However,
$$
q_\ell[u]=q_0[u]+\int_{\R}f(r)\frac{\ell(\ell+1)}{r^2}|u(x)|^2\dd x,
$$
and the last integral is strictly positive for every nonzero solution.
This is a contradiction.

If $\Delta_0<0$, the two eigenvalues of $I_2-B_0$ have opposite signs, so exactly one eigenvalue of $B_0$ is larger than $1$.
The other is smaller than $1$ because the trace of $B_0$ is $a_1+a_2<2$.
Strict monotonicity and the limit $B(\kappa)\to0$ therefore give exactly one crossing through $1$, and hence exactly one negative $s$-wave eigenvalue.
Solving $\Delta_0=0$ for $a_2$ gives Eq.~\eqref{eq:criticalcurve}; its denominator is positive because $a_1<1$ and $0<\chi<1$.
\end{proof}

The statement for $\Delta_0<0$ counts the $s$-wave modes.
It does not exclude a higher-angular-momentum instability farther inside the unstable region.
The theorem gives a strict collective effect because an open part of the square $0<a_1,a_2<1$ has $\Delta_0<0$.

For equal normalized strengths $a_1=a_2=a$, Eq.~\eqref{eq:Delta0} has two algebraic roots.
Only the smaller root lies below the one-shell threshold, and it is
\begin{equation}
a_{\rm c}=\frac{1}{1+\sqrt{\chi}}.
\label{eq:equalcritical}
\end{equation}
Consequently,
\begin{equation}
\frac{1}{2}<a_{\rm c}<1.
\label{eq:acbounds}
\end{equation}
Two nearly coincident shells need approximately one half of the isolated critical strength each.
For weak static overlap, each shell must approach its own threshold before the collective effect occurs.
The near-coincident limit has a simple interpretation.
When $R_2\downarrow R_1$, the two equal delta layers act as one layer with twice the strength, consistently giving $a_{\rm c}\to1/2$.
In the opposite weak-overlap limit $\chi\to0$, one obtains $a_{\rm c}\to1$, so the two shells recover their separate one-shell thresholds.

Let $c=(c_1,c_2)^{\mathrm T}$ be the vector of shell values, where the superscript $\mathrm T$ denotes transpose.
The vector $v=\Gamma^{1/2}c$ is the symmetrized Birman--Schwinger vector.
At equal normalized strength, the critical vector $v$ is proportional to $(1,1)^{\mathrm T}$.
Thus the first collective mode has the same sign on the two shells after Birman--Schwinger weighting.
It is the analog of the lower in-phase state of an attractive double well.

Figure~\ref{fig:phase} shows the critical curve for $R_1=4M$ and $R_2=8M$.
For this choice, $\chi=0.4150$ and $a_{\rm c}=0.6082$.
%%%
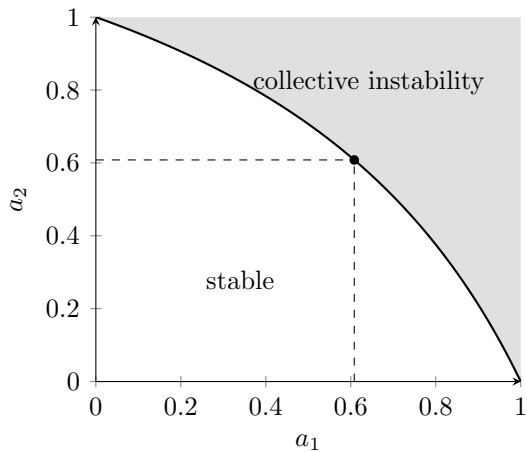
\begin{figure}[htbp]
\centering
\begin{tikzpicture}
\begin{axis}[
 width=7.2cm,height=6.4cm,
 xmin=0,xmax=1,ymin=0,ymax=1,
 xlabel={$a_1$},ylabel={$a_2$},
 axis lines=left,
 samples=180,
 tick label style={font=\small},
 label style={font=\small}
]
 \addplot[name path=curve,domain=0:1,black,thick]
 {(1-x)/(1-(1-0.4150374993)*x)};
 \path[name path=top] (axis cs:0,1)--(axis cs:1,1);
 \addplot[black!12] fill between[of=curve and top];
 \addplot[only marks,mark=*,mark size=1.6pt]
 coordinates {(0.6081859244,0.6081859244)};
 \draw[dashed] (axis cs:0.6081859244,0)--(axis cs:0.6081859244,0.6081859244);
 \draw[dashed] (axis cs:0,0.6081859244)--(axis cs:0.6081859244,0.6081859244);
 \node[font=\small,anchor=east] at (axis cs:0.94,0.82) {collective instability};
 \node[font=\small] at (axis cs:0.34,0.28) {stable};
\end{axis}
\end{tikzpicture}
\caption{The collective threshold in the square $0\leq a_1,a_2\leq1$ for $R_1=4M$ and $R_2=8M$.
The curve is Eq.~\eqref{eq:criticalcurve} with $\chi=0.4150$.
The shaded region has $\Delta_0<0$ and contains a growing $s$-wave mode.}
\label{fig:phase}
\end{figure}
%%%%%%%%%%%%%%%%%%%%%%%%%
\subsection{Closed form of the critical scalar cloud}

The determinant condition also determines the static scalar field at
the onset in elementary form.
Let $L(r)=-\log f(r)$ for $r>2M$.

\begin{proposition}[Critical scalar cloud]
Assume $0<a_1<1$, $0<a_2<1$, and $\Delta_0=0$.
Let
$$
c=
\left(
\begin{array}{c}
c_1\\
c_2
\end{array}
\right)
$$
be a positive vector satisfying $K_0(0)c=0$, and set
\begin{equation}
q_j=g_jc_j>0.
\label{eq:qjcritical}
\end{equation}
Up to an overall normalization, the critical static scalar field is
\begin{equation}
\varphi_{\rm c}(r)
=\left\{
\begin{array}{ll}
\displaystyle
\frac{q_1R_1L_1+q_2R_2L_2}{2M},
&2M<r<R_1,\\
\displaystyle
\frac{q_1R_1L(r)+q_2R_2L_2}{2M},
&R_1<r<R_2,\\
\displaystyle
\frac{(q_1R_1+q_2R_2)L(r)}{2M},
&r>R_2.
\end{array}
\right.
\label{eq:criticalprofile}
\end{equation}
It is regular at the future horizon, is strictly positive, and satisfies
\begin{equation}
\varphi_{\rm c}(r)
=\frac{Q_{\rm c}}{r}+O(r^{-2}),
\qquad
Q_{\rm c}=q_1R_1+q_2R_2>0.
\label{eq:criticalcharge}
\end{equation}
\end{proposition}

\begin{proof}
The boundary equation is
$$
c=\mathsf M_0(0)\Gamma c,
\qquad
\Gamma=\diag(g_1,g_2).
$$
After multiplication by $\Gamma^{1/2}$, the vector $v=\Gamma^{1/2}c$ satisfies $B_0v=v$.
At the critical point, $1$ is the simple largest eigenvalue of the matrix $B_0$, whose entries are strictly positive.
Its eigenvector can therefore be chosen with both components positive, and the same is true of $c=\Gamma^{-1/2}v$.
Hence $q_1,q_2>0$.

The radial function constructed from the threshold Green kernel is
$$
u_{\rm c}(x)=\sum_{j=1}^2q_jG_0^{\BH}(x,x_j;0).
$$
Substitution of Eq.~\eqref{eq:zerogreen} and division by $r$ give Eq.~\eqref{eq:criticalprofile} in the three radial regions.
The Green construction also gives continuity and the derivative jumps at both shells.
Since $L(r)>0$, the profile is strictly positive and has no node.
Finally, $L(r)=2M/r+O(r^{-2})$ gives Eq.~\eqref{eq:criticalcharge}.
\end{proof}

The radial function $u_{\rm c}=r\varphi_{\rm c}$ tends to a nonzero constant both as $x\to-\infty$ and as $x\to+\infty$.
It is therefore bounded but not square integrable.
The onset state is a zero-energy resonance rather than an $L^2$ eigenfunction.
When the critical curve is crossed toward $\Delta_0<0$, this resonance moves to a negative eigenvalue and becomes a square-integrable growing mode.

For equal normalized strengths, the symmetrized critical vector is proportional to $(1,1)^{\mathrm T}$.
The corresponding cloud is the weighted in-phase combination of the two shell responses.
Figure~\ref{fig:criticalprofile} shows the exact profile for $R_1=4M$, $R_2=8M$, and $a_1=a_2=a_{\rm c}$.

\begin{figure}[htbp]
\centering
\begin{tikzpicture}
\begin{axis}[
 width=10.0cm,height=6.1cm,
 xmin=2,xmax=20,ymin=0.2,ymax=1.12,
 xlabel={$r/M$},ylabel={$\varphi_{\rm c}(r)/\varphi_{\rm c}(2M)$},
 axis lines=left,
 tick label style={font=\small},
 label style={font=\small}
]
 \addplot[black,thick,domain=2.01:4,samples=2] {1};
 \addplot[black,thick,domain=4:8,samples=100]
 {((4/sqrt(5.5451774445))*(-ln(1-2/x))+(8/sqrt(9.2058263185))*(-ln(0.75)))/2)/
  (((4/sqrt(5.5451774445))*(-ln(0.5))+(8/sqrt(9.2058263185))*(-ln(0.75)))/2)};
 \addplot[black,thick,domain=8:20,samples=120]
 {(((4/sqrt(5.5451774445))+(8/sqrt(9.2058263185)))*(-ln(1-2/x))/2)/
  (((4/sqrt(5.5451774445))*(-ln(0.5))+(8/sqrt(9.2058263185))*(-ln(0.75)))/2)};
 \draw[dashed] (axis cs:4,0.2)--(axis cs:4,1.05);
 \draw[dashed] (axis cs:8,0.2)--(axis cs:8,1.05);
 \node[font=\small] at (axis cs:4,1.055) {$R_1$};
 \node[font=\small] at (axis cs:8,1.055) {$R_2$};
\end{axis}
\end{tikzpicture}
\caption{The normalized critical scalar cloud from Eq.~\eqref{eq:criticalprofile} for $R_1=4M$, $R_2=8M$, and equal normalized strengths.
The field is constant between the horizon and the inner shell, changes its radial slope at each shell, and decays as $1/r$ outside the outer shell.}
\label{fig:criticalprofile}
\end{figure}
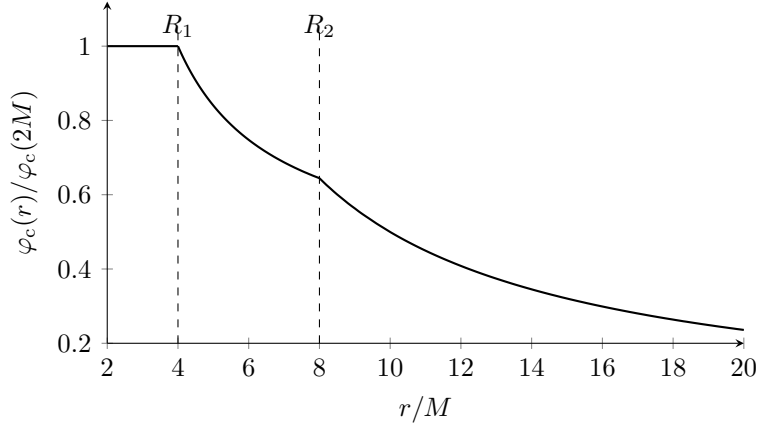

\subsection{Relation to shell mass, pressure, and the probe limit}
\label{subsec:physical-parameters}

The normalized strength has a direct expression in terms of the surface trace.
From Eqs.~\eqref{eq:singularmass}, \eqref{eq:G11zero}, and \eqref{eq:aj}, one obtains
\begin{equation}
a_j
=\frac{2\pi\beta\tau_jR_j^2\sqrt{f(R_j)}L_j}{M}
\label{eq:atrace}
\end{equation}
under the attractive sign $\beta\tau_j>0$.
For a pressureless surface source with $\beta<0$, the total proper energy is
$$
m_j=4\pi R_j^2\sigma_j.
$$
Then
\begin{equation}
a_j
=\frac{|\beta|m_j}{2M}\sqrt{f(R_j)}L_j.
\label{eq:amass}
\end{equation}
Equation~\eqref{eq:amass} separates the theory coupling, the shell mass ratio, and the redshift geometry.
At fixed normalized strength it gives
\begin{equation}
\frac{|\beta|m_j}{M}
=\frac{2a_j}{\sqrt{f(R_j)}L_j}.
\label{eq:betamass}
\end{equation}
For $R_j=O(M)$ away from an extremely near-horizon position, the right side is of order one.
Thus a shell mass ratio much smaller than one requires a correspondingly large effective coupling if a pressureless source is to reach the onset while the metric remains close to Schwarzschild.
Surface pressure can change this estimate because the relevant quantity is $\tau_j=-\sigma_j+2p_j$, not the rest mass alone.

We now compare this prescribed-source conversion with a static shell satisfying the Israel conditions.
Consider first one static shell at radius $R$ between Schwarzschild regions with masses $M_-$ and $M_+=M_-+\delta M$.
Writing $f_\pm=1-2M_\pm/R$, the Israel equations give
\begin{eqnarray}
\sigma&=&\frac{\sqrt{f_-}-\sqrt{f_+}}{4\pi R},
\label{eq:israel-sigma}\\
p&=&\frac{1}{8\pi R}
\left\{
\frac{1-M_+/R}{\sqrt{f_+}}
-\frac{1-M_-/R}{\sqrt{f_-}}
\right\}.
\label{eq:israel-pressure}
\end{eqnarray}
For $\delta M/M\ll1$, with $M_-=M$ at leading order, these relations become
\begin{eqnarray}
\sigma&=&\frac{\delta M}{4\pi R^2\sqrt{f(R)}}
+O\left(\frac{\delta M^2}{R^3}\right),
\nonumber\\
p&=&\frac{\delta M M}{8\pi R^3f(R)^{3/2}}
+O\left(\frac{\delta M^2}{R^3}\right).
\label{eq:israel-probe-expansion}
\end{eqnarray}
Consequently, if $m=4\pi R^2\sigma$ denotes the leading proper shell energy,
\begin{equation}
\frac{p}{\sigma}
=\frac{M}{2Rf(R)}+O\left(\frac{m}{M}\right),
\qquad
\tau
=-\sigma\left\{1-\frac{M}{Rf(R)}\right\}
+O\left(\frac{m^2}{M R^2}\right).
\label{eq:israel-trace-expansion}
\end{equation}
For two light shells, the mass already carried by the other shell changes these formulas only at the next order in the total shell mass ratio.
Thus Eq.~\eqref{eq:israel-trace-expansion} may be applied to each shell at the order retained in the present probe calculation.

For the standard sign $\beta<0$, such a static shell is attractive when $R>3M$, is trace free at $R=3M$, and gives a repulsive scalar delta term for $2M<R<3M$.
Thus the exact junction condition in Eq.~\eqref{eq:exact-trace-zero} reduces, in the leading small-mass limit, to a trace-sign transition exactly at the Schwarzschild photon sphere.
There is also a physically different attractive regime.
For
$$
\frac{5M}{2}\leq R<3M,
$$
Eq.~\eqref{eq:israel-trace-expansion} gives $1/2<p/\sigma\leq1$.
Hence the surface dominant energy condition $\sigma\geq|p|$ is satisfied, while $\tau>0$ and a positive coupling $\beta>0$ produces an attractive scalar interaction.
For $R_j>3M$, substitution into Eq.~\eqref{eq:atrace} gives the leading mechanically supported-shell conversion
\begin{equation}
a_j
=\frac{|\beta|m_j}{2M}
\left\{1-\frac{M}{R_jf(R_j)}\right\}
\sqrt{f(R_j)}L_j
+O\left(|\beta|\epsilon_{\rm sh}^2\right).
\label{eq:aisraelmass}
\end{equation}
Equation~\eqref{eq:amass} is used only to convert the surface trace into a mass scale.
It does not describe a static unsupported dust shell.

For example, take the equal-strength threshold for $R_1=4M$ and $R_2=8M$.
Equations~\eqref{eq:equalcritical} and \eqref{eq:betamass} give
\begin{equation}
a_{\rm c}=0.608186,
\qquad
\frac{|\beta|m_1}{M}=2.482,
\qquad
\frac{|\beta|m_2}{M}=4.882.
\label{eq:physical-example}
\end{equation}
Thus, if the proper energy of each pressureless source is one percent of the black-hole mass, the critical values inferred separately from the two shell locations are of order $|\beta|=248$ and $|\beta|=488$.
These are two separate conversions; equal shell masses with a common
value of $\beta$ do not give equal normalized strengths.
The normalized collective reduction is substantial, but the coupling required by a light pressureless source can still be large.
These numbers are the pressureless-source conversion.
If instead the leading static Israel pressure in Eq.~\eqref{eq:israel-trace-expansion} is imposed, the corresponding products become
\begin{equation}
\frac{|\beta|m_1}{M}=4.964,
\qquad
\frac{|\beta|m_2}{M}=5.859
\label{eq:physical-israel-example}
\end{equation}
for $R_1=4M$ and $R_2=8M$.
This example is not an observational constraint on $\beta$ because the thin-shell model does not specify an astrophysical matter configuration.
It instead quantifies the physical cost of the double scaling used here.
More generally, the mass estimate alone is insufficient: the instability is controlled by the full surface trace.

A phase diagram in physical coupling--mass variables makes this distinction explicit.
Set
$$
\mu_j=\frac{|\beta|m_j}{M},
\qquad
a_j=c_j^{\rm I}\mu_j,
\qquad
c_j^{\rm I}
=\frac{1}{2}
\left(1-\frac{M}{R_jf(R_j)}\right)\sqrt{f(R_j)}L_j,
$$
for the leading static Israel conversion with $R_j>3M$.
For $R_1=4M$ and $R_2=8M$,
$$
c_1^{\rm I}=0.122532,
\qquad
c_2^{\rm I}=0.103808.
$$
Figure~\ref{fig:physical-phase} shows Eq.~\eqref{eq:criticalcurve} in the variables $(\mu_1,\mu_2)$.
The isolated thresholds are $\mu_{1,{\rm c}}=8.161$ and $\mu_{2,{\rm c}}=9.633$, while the equal-normalized-strength collective point is $(4.963,5.859)$.

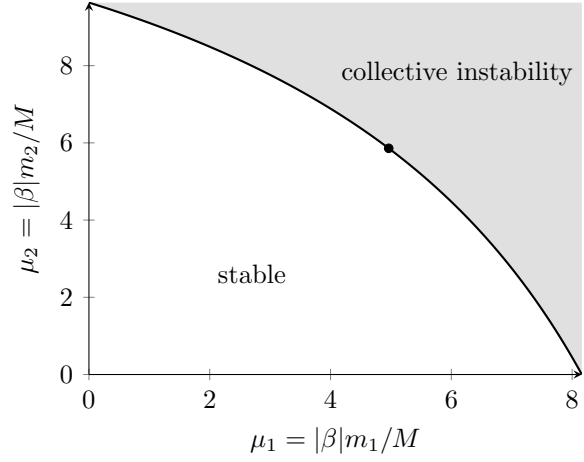
\begin{figure}[htbp]
\centering
\begin{tikzpicture}
\begin{axis}[
 width=8.1cm,height=6.5cm,
 xmin=0,xmax=8.1612,ymin=0,ymax=9.6333,
 xlabel={$\mu_1=|\beta|m_1/M$},
 ylabel={$\mu_2=|\beta|m_2/M$},
 axis lines=left,
 samples=180,
 tick label style={font=\small},
 label style={font=\small}
]
 \addplot[name path=physcurve,domain=0:8.16111557,black,thick]
 {(1-0.1225322679*x)/(0.1038083262*(1-(1-0.4150374993)*0.1225322679*x))};
 \path[name path=phystop] (axis cs:0,9.63313865)--(axis cs:8.16111557,9.63313865);
 \addplot[black!12] fill between[of=physcurve and phystop];
 \addplot[only marks,mark=*,mark size=1.6pt]
 coordinates {(4.96347562,5.85873934)};
 \node[font=\small] at (axis cs:6.1,7.8) {collective instability};
 \node[font=\small] at (axis cs:2.7,2.6) {stable};
\end{axis}
\end{tikzpicture}
\caption{The collective threshold in the physical coupling--mass variables for leading static Israel shells at $R_1=4M$ and $R_2=8M$.
The shaded region is unstable in the $s$-wave sector.
The marked point has equal normalized strengths $a_1=a_2=a_{\rm c}$, not equal coupling--mass products.}
\label{fig:physical-phase}
\end{figure}

The exact threshold is a statement about the linear operator for a prescribed integrated trace.
It is not by itself a claim that a given astrophysical shell and a given observationally allowed scalar-tensor theory lie in the unstable region.
A self-consistent model must solve the shell mechanics and the junction conditions together with the scalar field.

\section{Representative thresholds, numerical check, and physical scope}
\label{sec:representative-thresholds}

\subsection{Representative thresholds}
\label{subsec:representative-thresholds}

Table~\ref{tab:thresholds} gives the geometric overlap and the equal-strength collective threshold for several shell radii.
No radial integration is required for these values.
They follow directly from Eqs.~\eqref{eq:chi} and \eqref{eq:equalcritical}.

\begin{table}[H]
\centering
\caption{Equal normalized collective thresholds and the corresponding pressureless-source coupling--mass products.
The last two columns give the diagnostic conversion $|\beta|m_j/M$ from Eq.~\eqref{eq:betamass} at $a_1=a_2=a_{\rm c}$.
A mechanically supported static shell has the pressure correction in Eq.~\eqref{eq:aisraelmass}. All listed inner radii exceed $3M$, so the leading static trace has the standard attractive sign for $\beta<0$.}
\label{tab:thresholds}
\begin{tabular}{cccccc}
\toprule
$R_1/M$ & $R_2/M$ & $\chi$ & $a_{\rm c}$
& $|\beta|m_1/M$ & $|\beta|m_2/M$\\
\midrule
$3.5$ & $6$ & $0.4785$ & $0.5911$ & $2.131$ & $3.571$\\
$4$ & $6$ & $0.5850$ & $0.5666$ & $2.312$ & $3.423$\\
$4$ & $8$ & $0.4150$ & $0.6082$ & $2.482$ & $4.882$\\
$4$ & $12$ & $0.2630$ & $0.6610$ & $2.697$ & $7.943$\\
$6$ & $10$ & $0.5503$ & $0.5741$ & $3.468$ & $5.753$\\
\bottomrule
\end{tabular}
\end{table}

The values show that the collective reduction can be large in normalized units.
For example, shells at $4M$ and $6M$ need only about $56.7$ percent of their separate critical strengths when the normalized strengths are equal.
The reduction becomes smaller when the outer shell is moved far away, because the static scalar overlap decreases.
The last two columns also show that moving the outer shell farther away raises the coupling--mass product needed to keep it at the same normalized threshold.
Thus a large reduction of $a_{\rm c}$ does not by itself imply a small
value of $|\beta|m_j/M$.

\subsection{Numerical check of the finite growth rate}
\label{subsec:numerical-growth-rate}

As a numerical check, we calculate the finite growth rate beyond the analytic threshold.

We use
$$
M=1,
\qquad
R_1=4M,
\qquad
R_2=8M,
\qquad
a_1=a_2=a.
$$
For these radii,
$$
\chi=\frac{\log f(R_2)}{\log f(R_1)}=0.415037,
\qquad
a_{\rm c}=\frac{1}{1+\sqrt{\chi}}=0.608186.
$$
For each $a$, we solve
$$
D_0(\kappa)
=\left(1-g_1G_{11}(\kappa)\right)
\left(1-g_2G_{22}(\kappa)\right)
-g_1g_2G_{12}(\kappa)^2=0,
$$
where
$$
G_{ij}(\kappa)=G_0(x_i,x_j;-\kappa^2),
\qquad
g_j=\frac{a}{G_{jj}^{(0)}}.
$$
The left solution of $u''=(V_0^{\BH}+\kappa^2)u$ is integrated from the horizon-side cutoff with $u'/u=\kappa$.
The right solution is integrated from the infinity-side cutoff with $u'/u=-\kappa$.
Here the prime denotes differentiation with respect to $x$.
The Green matrix is then obtained from Eq.~\eqref{eq:greennegative}.
For the plotted calculation, the radial cutoffs were $r_{\rm L}=2M+10^{-7}M$ and $r_{\rm R}=\max\{2000M,4/\kappa\}$.
The integrations were performed with a DOP853 adaptive eighth-order Runge--Kutta method, using relative tolerance $2\times10^{-10}$ and absolute tolerance $2\times10^{-12}$.
Replacing $4/\kappa$ by $6/\kappa$ changed the displayed roots by less than $10^{-9}/M$.
The positive zero of the determinant was found by Brent's bracketed root method, with the exact value of $D_0(0)$ used at the lower endpoint.
For example,
$$
\kappa M=0.0183522\quad\hbox{at }a=0.70,
\qquad
\kappa M=0.0492376\quad\hbox{at }a=0.90.
$$

Figure~\ref{fig:finite-growth-rate} shows the result.

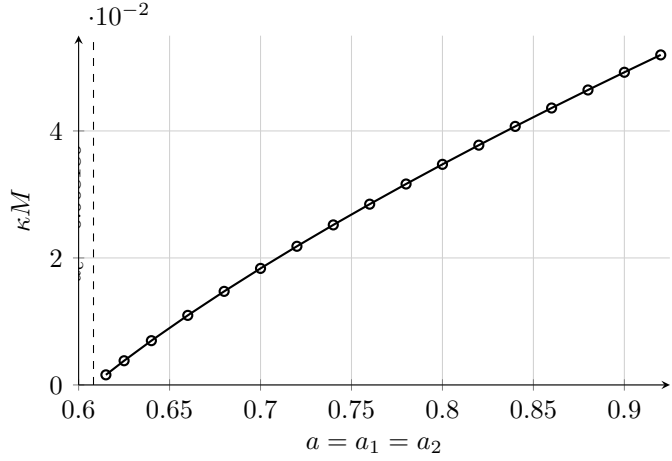
\begin{figure}[tb]
\centering
\begin{tikzpicture}
\begin{axis}[
 width=9.4cm,height=6.2cm,
 xmin=0.60,xmax=0.925,ymin=0,ymax=0.055,
 xlabel={$a=a_1=a_2$},ylabel={$\kappa M$},
 axis lines=left,
 grid=both,
 major grid style={black!18},
 minor grid style={black!8},
 tick label style={font=\small},
 label style={font=\small}
]
 \addplot[black,thick,mark=o,mark size=1.7pt]
 coordinates {
 (0.615,0.0015886693)
 (0.625,0.0038097548)
 (0.640,0.0069713443)
 (0.660,0.0109542934)
 (0.680,0.0147342982)
 (0.700,0.0183522139)
 (0.720,0.0218349554)
 (0.740,0.0252018157)
 (0.760,0.0284673924)
 (0.780,0.0316431594)
 (0.800,0.0347383945)
 (0.820,0.0377607673)
 (0.840,0.0407167306)
 (0.860,0.0436117928)
 (0.880,0.0464507134)
 (0.900,0.0492376466)
 (0.920,0.0519762513)
 };
 \draw[dashed] (axis cs:0.6081859244,0)--(axis cs:0.6081859244,0.055);
 \node[rotate=90,font=\scriptsize,anchor=south] at (axis cs:0.6081859244,0.027)
 {$a_{\rm c}=0.608186$};
\end{axis}
\end{tikzpicture}
\caption{Dimensionless growth rate $\kappa M$ of the unstable $s$-wave mode for
$R_1=4M$, $R_2=8M$, and $a_1=a_2=a$.
The vertical dashed line indicates the exact analytic threshold
$a_{\rm c}=0.608186$.
A positive root of $D_0(\kappa)=0$ appears for $a>a_{\rm c}$ and approaches zero as
$a\downarrow a_{\rm c}$.}
\label{fig:finite-growth-rate}
\end{figure}

No positive root was found below $a_{\rm c}$, while one positive root
appears above it and tends to zero at the threshold.
This agrees with Theorem~\ref{thm:collective-threshold}.

\subsection{Physical scope}
\label{subsec:physical-scope}

The model has the following limitations.
First, the analysis concerns the linear scalar onset around a scalar-free branch.
It does not establish the existence or stability of a nonlinear two-shell scalarized branch.
Second, Sec.~\ref{sec:exact-background} constructs the exact scalar-free geometry and the surface stresses required by the two Israel junctions, but it does not provide a microscopic material equation of state or establish radial mechanical stability.
The exact junction calculation and the scalar determinant play different roles: the former fixes the background and surface tensor, while the latter determines whether that scalar-free background has a tachyonic scalar mode.
Third, the model is spherical.
An accretion disk or a torus has angular structure and couples different partial waves.
The three-dimensional boundary operator for nonconstant shell strengths in Ref.~\cite{Kaminaga2026} gives a possible starting point for that problem.
Fourth, rotation is not included.
For Kerr black holes, the growing-mode and superradiant channels can coexist \cite{Cardoso2013PRD}, and the simple angular momentum ordering used here is lost.
Finally, the large negative values of $\beta$ required in the standard-sign light-shell examples should not be interpreted as phenomenologically viable in the original massless Damour--Esposito-Far\`ese model, whose negative-$\beta$ region is tightly constrained by binary-pulsar observations \cite{Doneva2024}. The phenomenological status of the positive-$\beta$ regime discussed above is different and is not assessed here.
The result applies to other scalar-matter theories only when their
weak-field constraints permit the required coupling.

The delta model should be read as a thin-layer limit with fixed integrated trace, rather than as the claim that realistic matter has zero thickness.
The elementary resistance kernel in Sec.~\ref{sec:exact-background} assumes that the bulk scalar effective-mass term vanishes away from the shell surfaces.
For finite-width layers, one must solve the zero-frequency radial equation with the corresponding smooth effective-mass profile.
The delta-shell formula should be recovered when $w_j$ is short compared with the radial scale of the critical cloud and when
$$
\int T_j(\zeta)\dd\zeta=\tau_j
$$
is kept fixed.
A numerical comparison with smooth layers would test this thin-layer limit.
No finite-width result is used in the present derivation.
The closed profile in Eq.~\eqref{eq:criticalprofile} can be used as the reference solution for this comparison.

\section{Conclusion}

We have derived the collective linear onset of matter-induced scalarization for a black hole surrounded by two thin matter shells.
The scalar perturbation decouples at first order because $\Phi_0=0$ and $\alpha(0)=0$.
The shell source is fixed by the invariant surface trace, not by an arbitrary potential parameter.

For a general static spherical black-hole exterior with no bulk scalar effective-mass term in the open regions between the shells, the relevant geometric quantity is the scalar resistance
$$
{\cal S}(R)=\int_R^\infty
\frac{\dd r}{r^2\sqrt{h(r)f(r)}}.
$$
The exact rank-two onset condition is
$$
1-a_1-a_2+(1-\chi)a_1a_2=0,
\qquad
\chi=\frac{{\cal S}(R_2)}{{\cal S}(R_1)}.
$$
This formula applies both to the Schwarzschild probe problem and to the exact scalar-free spacetime formed by three Schwarzschild regions joined by two Israel shells.
The explicit comparison in Table~\ref{tab:exact-probe-comparison} shows that the probe threshold receives the expected linear correction in the shell mass ratio.
Thus the collective determinant is not an artifact of omitting metric backreaction.

In the Schwarzschild limit,
$$
{\cal S}(R)=-\frac{1}{2M}\log\left(1-\frac{2M}{R}\right),
$$
so the geometric overlap is
$$
\chi=\frac{\log f(R_2)}{\log f(R_1)}.
$$
Two shells with $a_1<1$ and $a_2<1$ can therefore destabilize the scalar-free black hole as a pair.
For equal normalized strengths, the threshold is $a_{\rm c}=1/(1+\sqrt{\chi})$.
The critical scalar cloud is positive, regular at the horizon, and decays as $1/r$.
It is a zero-energy resonance that becomes a square-integrable growing mode beyond threshold, as confirmed by the finite-$\kappa$ calculation.

The Israel-shell conversion gives two additional physical conclusions.
First, a strong reduction of the normalized threshold need not imply a moderate coupling for a light matter layer.
Second, the exact Israel trace changes sign, for a nontrivial shell, at $\sqrt{f_-(R)f_+(R)}=1/3$.
Its light-shell limit is $R=3M$, the Schwarzschild photon sphere.
For $5M/2\leq R<3M$ in that limit, a static shell can satisfy the surface dominant energy condition while having positive trace, so a positive scalar-matter coupling becomes attractive.

The present paper determines the onset and its geometric origin.
The next problem is to continue the threshold cloud to a nonlinear scalarized two-shell branch and to combine that continuation with a material equation of state and radial-stability analysis.
Possible extensions include finite-width layers, nonspherical matter, and rotation.
These problems require the corresponding generalization of the static Green function and the boundary operator.

\section*{Statements and Declarations}

\subsection*{Competing Interests}
The author declares that he has no competing interests.

\subsection*{Funding}
No funding was received for conducting this study.

\subsection*{Data Availability}
No external datasets were used in this study.
The numerical results can be reproduced from the equations, algorithms,
tolerances, and parameters given in the paper.

\subsection*{Author Contributions}
The author contributed to all parts of the manuscript.


\begin{thebibliography}{99}

\bibitem{Damour1993}
T. Damour and G. Esposito-Far\`ese,
Nonperturbative strong-field effects in tensor-scalar theories of gravitation,
Phys. Rev. Lett. 70 (1993) 2220--2223.
\url{https://doi.org/10.1103/PhysRevLett.70.2220}.

\bibitem{Damour1996}
T. Damour and G. Esposito-Far\`ese,
Tensor-scalar gravity and binary-pulsar experiments,
Phys. Rev. D 54 (1996) 1474--1491.
\url{https://doi.org/10.1103/PhysRevD.54.1474}.

\bibitem{Harada1997}
T. Harada,
Stability analysis of spherically symmetric star in scalar-tensor theories of gravity,
Prog. Theor. Phys. 98 (1997) 359--379.
\url{https://doi.org/10.1143/PTP.98.359}.

\bibitem{Doneva2024}
D. D. Doneva, F. M. Ramazano\u{g}lu, H. O. Silva, T. P. Sotiriou, and S. S. Yazadjiev,
Spontaneous scalarization,
Rev. Mod. Phys. 96 (2024) 015004.
\url{https://doi.org/10.1103/RevModPhys.96.015004}.

\bibitem{Sotiriou2012}
T. P. Sotiriou and V. Faraoni,
Black holes in scalar-tensor gravity,
Phys. Rev. Lett. 108 (2012) 081103.
\url{https://doi.org/10.1103/PhysRevLett.108.081103}.

\bibitem{Cardoso2013PRD}
V. Cardoso, I. P. Carucci, P. Pani, and T. P. Sotiriou,
Matter around Kerr black holes in scalar-tensor theories: scalarization and superradiant instability,
Phys. Rev. D 88 (2013) 044056.
\url{https://doi.org/10.1103/PhysRevD.88.044056}.

\bibitem{Cardoso2013PRL}
V. Cardoso, I. P. Carucci, P. Pani, and T. P. Sotiriou,
Black holes with surrounding matter in scalar-tensor theories,
Phys. Rev. Lett. 111 (2013) 111101.
\url{https://doi.org/10.1103/PhysRevLett.111.111101}.

\bibitem{Tanaka2025}
J. Tanaka,
Scalarization and superradiant instability of black hole induced by dark matter halo in the scalar-tensor theory of gravity,
Phys. Rev. D 112 (2025) 064027.
\url{https://doi.org/10.1103/czh8-8gfh}.

\bibitem{Sadjadi2026}
H. Mohseni Sadjadi and M. Navid Gasemi Zad,
Mass-varying dark matter induced scalarization and scalar clouds around black holes,
arXiv:2606.29864 [gr-qc] (2026).

\bibitem{GasemiZad2026}
M. Navid Gasemi Zad and H. Mohseni Sadjadi,
Spontaneous scalarization around a black hole in a dark matter halo,
arXiv:2607.17367 [gr-qc] (2026).

\bibitem{Laeuger2025}
A. Laeuger, C. Weller, D. Li, and Y. Chen,
Ringdown of a black hole surrounded by a thin shell of matter,
Phys. Rev. D 112 (2025) 084042.
\url{https://doi.org/10.1103/gkj4-m1c1}.

\bibitem{KaminagaEcho2026}
M. Kaminaga,
Black-hole echo resonance spectra and source dependence in a controlled transfer-function model,
Int. J. Mod. Phys. D 35 (2026) 2650053.
\url{https://doi.org/10.1142/S0218271826500537}.

\bibitem{Israel1966}
W. Israel,
Singular hypersurfaces and thin shells in general relativity,
Nuovo Cimento B 44 (1966) 1--14;
Erratum, Nuovo Cimento B 48 (1967) 463.
\url{https://doi.org/10.1007/BF02710419};
Erratum: \url{https://doi.org/10.1007/BF02712210}.

\bibitem{AcunaCardenas2024}
R. O. Acu\~na-C\'ardenas, O. Sarbach, and L. Tessieri,
Wave propagation through a spacetime containing thin concentric shells of matter,
Phys. Rev. D 110 (2024) 104064.
\url{https://doi.org/10.1103/PhysRevD.110.104064}.

\bibitem{Cardoso2020}
V. Cardoso, A. Foschi, and M. Zilh\~ao,
Collective scalarization or tachyonization: when averaging fails,
Phys. Rev. Lett. 124 (2020) 221104.
\url{https://doi.org/10.1103/PhysRevLett.124.221104}.

\bibitem{Albeverio2005}
S. Albeverio, F. Gesztesy, R. H\o egh-Krohn, and H. Holden,
Solvable Models in Quantum Mechanics, second edition,
AMS Chelsea Publishing, Providence, 2005.

\bibitem{Antoine1987}
J.-P. Antoine, F. Gesztesy, and J. Shabani, 
Exactly solvable models of sphere interactions in quantum mechanics, 
J. Phys. A: Math. Gen. 20 (1987) 3687--3712. 
\url{https://doi.org/10.1088/0305-4470/20/12/022}.

\bibitem{Shabani1988}
J. Shabani,
Finitely many delta interactions with supports on concentric spheres,
J. Math. Phys. 29 (1988) 660--664.
\url{https://doi.org/10.1063/1.528005}.

\bibitem{Behrndt2013}
J. Behrndt, M. Langer, and V. Lotoreichik,
Schr\"odinger operators with $\delta$ and $\delta'$-potentials supported on hypersurfaces,
Ann. Henri Poincar\'e 14 (2013) 385--423.
\url{https://doi.org/10.1007/s00023-012-0189-5}.

\bibitem{Hounkonnou1997}
M. N. Hounkonnou, M. Hounkpe, and J. Shabani,
Scattering theory for finitely many sphere interactions supported by concentric spheres,
J. Math. Phys. 38 (1997) 2832--2850.
\url{https://doi.org/10.1063/1.532022}.

\bibitem{Kaminaga2026}
M. Kaminaga,
Schr\"odinger operators with concentric $\delta$-shell interactions,
Anal. Math. Phys. 16 (2026), Article 68.
\url{https://doi.org/10.1007/s13324-026-01191-w}.

\bibitem{Berti2009}
E. Berti, V. Cardoso, and A. O. Starinets,
Quasinormal modes of black holes and black branes,
Class. Quantum Grav. 26 (2009) 163001.
\url{https://doi.org/10.1088/0264-9381/26/16/163001}.

\end{thebibliography}
\end{document}